%% file: paper-arxiv.tex
\documentclass[format=acmsmall, review=false]{acmart}

\makeatletter
\def\@ACM@checkaffil{
    \if@ACM@instpresent\else
    \ClassWarningNoLine{\@classname}{No institution present for an affiliation}%
    \fi
    \if@ACM@citypresent\else
    \ClassWarningNoLine{\@classname}{No city present for an affiliation}%
    \fi
    \if@ACM@countrypresent\else
        \ClassWarningNoLine{\@classname}{No country present for an affiliation}%
    \fi
}
\makeatother

\usepackage{acm-ec-23}
\usepackage{booktabs} 
\usepackage[ruled]{algorithm2e} 

\SetAlFnt{\small}
\SetAlCapFnt{\small}
\SetAlCapNameFnt{\small}
\SetAlCapHSkip{0pt}
\IncMargin{-\parindent}

\usepackage{amsfonts,amsmath,amsthm}
\newtheorem{theorem}{Theorem}
\newtheorem{definition}{Definition}

\newtheorem{lemma}{Lemma}

\newtheorem{example}{Example}

\usepackage{adjustbox}
\usepackage{rotating}
\usepackage{threeparttable}
\usepackage{booktabs}

\newcommand\fnote[1]{\captionsetup{font=scriptsize, justification=justified}\caption*{#1}}
\ExplSyntaxOn
\cs_new:Npn \expandableinput #1
  { \use:c { @@input } { \file_full_name:n {#1} } }
\ExplSyntaxOff

\setcitestyle{acmnumeric}

\title[A Task-Interdependency Model of Complex Collaboration]{A Task-Interdependency Model Of Complex Collaboration\\ Towards Human-Centered Crowd Work}

\author{David T. Lee}
\email{dlee105@ucsc.edu}
\affiliation{%
  \institution{University of California, Santa Cruz}
  \city{}
  \country{USA}
}

\author{Christos A. Makridis}
\email{cmakridi@stanford.edu}
\affiliation{%
  \institution{National Artificial Intelligence Institute and Stanford University}
  \city{}
  \country{USA}
}

\begin{abstract}
Models of crowdsourcing and human computation often assume that individuals independently carry out small, modular tasks. However, while these models have successfully shown how crowds can accomplish significant objectives, they can inadvertently advance a less than human view of crowd workers and fail to capture the unique human capacity for complex collaborative work. We present a model centered on interdependencies—a phenomenon well understood to be at the core of collaboration—that allows one to formally reason about diverse challenges to complex collaboration. Our model represents tasks as an interdependent collection of subtasks, formalized as a task graph. We use it to explain challenges to scaling complex collaborative work, underscore the importance of expert workers, reveal critical factors for learning on the job, and explore the relationship between coordination intensity and occupational wages. Using data from O*NET and the Bureau of Labor Statistics, we introduce an index of occupational coordination intensity to validate our theoretical predictions. We present preliminary evidence that occupations with greater coordination intensity are less exposed to displacement by AI, and discuss opportunities for models that emphasize the collaborative capacities of human workers, bridge models of crowd work and traditional work, and promote AI in roles augmenting human collaboration.
\end{abstract}

\begin{document}
\begin{titlepage}
\maketitle
\end{titlepage}

\section{Introduction and Related Work}

\noindent Mathematical models of crowdsourcing and human computation today largely assume small modular tasks, ``computational primitives'' such as labels, comparisons, or votes requiring little coordination~\cite{Little2010-dx,Chen2016-fw}. These have been used to understand how crowds can be organized to advance scientific discoveries, mobilize collective action, and contribute to participatory governance~\cite{Simpson2014-wx,Ghonim2012-qq,Pennycook2019-li,Groh2022-yr,Okolloh2009-gb,Noveck2015-lg,Michelucci2016-wo}.

However, this narrow micro-task view of crowd work has two significant drawbacks. First, it reinforces a less than human view of workers common in crowdsourcing, where workers are treated as low skilled, replaceable, and untrustworthy, carrying out simple tasks in online labor markets for low pay under algorithmic management~\citep{Gray2019-us,Horton2010-aq}, with requesters able to refuse pay at whim, and workers lacking opportunities for education and career growth~\citep{Kittur2013-qr,Rivera2021}. Could our micro-task models be \textit{``creating contexts in which people will be seduced by the economically convenient fiction that [crowd workers] are machines and should be treated as such?''}~\citep{Silberman2010-tk}. This need for more human-centric models is particularly pressing in light of advances in AI, the increasing use of human-in-the-loop methods, and calls to design AI to complement rather than only substitute human labor~\cite{Cranshaw2017-et,Barowy2012-fg,Brynjolfsson2014-xw,Heer2019-uf,Anthes2017-gl,Kittur2019-cc}.

Second, a narrow micro-task view of crowd work also fails to capture the human capacity for complex collaborative work where the main concerns are how to effectively structure, delegate, and collaborate on work that may be large in scope, underdefined, and highly interdependent~\cite{Valentine2017-fs,Kim2014-dp,Kulkarni2012-jw}. 
Citizen science projects have noted the centrality of modularity and interdependence for understanding the role of citizen scientists, experts, and computational technologies, and raised the need for methods to systematically create successful decompositions, reason about the impact of AI on citizen scientists, and support upskilling of participants~\cite{Ponti2022-tt,Franzoni2014-fs,Sauermann2015-rz,Trouille2019-lc}. 
As stated in a challenge problem for the field, how might we define new models for crowd work that 
\textit{``solve tasks that are much more challenging and less `transactional', complex problems where there are strong constraints and interdependencies''}~\cite{Chen2016-fw}.

This paper introduces a new model of complex crowd work that captures the unique human capacity to collaborate with others and the challenges humans face when doing so. 
Within AI, a large literature has studied mathematical models of micro-task crowd work, e.g. for optimizing or incentivizing quality, efficiency and cost~\cite{Khetan2016-sh,Sheng2008-fa,Mao2013-fz}, or for eliciting preferences for various goals~\cite{Lee2014-re,Benade2017-il,Yi2012-hi}. There is little work on mathematical models for \textit{complex crowd work} that might help one to reason about how AI can augment or complement human labor~\cite{Brynjolfsson2014-xw}. Some work has studied models of macro-task work, such as for studying the challenge of scheduling multiple workers to carry out multiple large tasks~\cite{Schmitz2018-te,Haas2015-oo}. However, for their focus, workers are modeled as contributing additively. Our interest is in modeling the interdependencies that make collaborative work not additive, and that make it difficult to divide up and collaborate on complex work, and have been the subject of study in numerous empirical studies on task complexity~\cite{Straub2023-tf,Almaatouq2021-lo,Yang2022-bv,Hulikal_Muralidhar2022-fy,Lustig2020-zv}.

In our model, tasks are represented as collections of interdependent subtasks, formalized as a task graph. Each node is a subtask with an arbitrary size parameter. Interdependencies, represented as node and edge weights, impose costs on workers who need to spend time absorbing context of relevant work. Importantly, workers do not have to pay this context cost for work they did themselves. 

To illustrate how this simple model can be used to reason about diverse aspects of complex collaboration, we apply it across several settings. We examine the limits of scaling complex crowd work even when one has infinite workers, showing how high interdependencies and low task granularity bound work capacity to a constant factor of the contributions of top workers, which is in turn limited when workers are short-term novices. This explains the evolution of complex crowd work from parallelizable tasks to modular workflows to long-term experts~\cite{Valentine2017-fs}, and answers questions on the limits of social-media based mobilization~\cite{Cebrian2016-oz}. 

We then examine recruitment and upskilling, showing the outsized role top workers play in determining work capacity. We define a stylized model of legimitate peripheral participation (LPP)~\cite{Lave1991-ds} that allows us to formally reason about its limitations. We show that decomposing tasks and eliminating interdependencies is not enough for successful LPP. Instead, appropriate patterns of interdependence are actually critical for LPP to facilitate pathways from novice tasks to expert ones. 

Finally, we turn to the economy as a setting where complex collaborative work already exists, using our model to explore the relationship between coordination intensity and occupational wages. Using occupational data from O*NET and the Bureau of Labor Statistics, we introduce a new index of coordination intensity and validate the predicted positive correlation. \footnote{Although our results are robust to controls for demographic characteristics and skill intensities across occupations, we recognize that they are correlations, rather than causal estimates.} 
We conclude with further empirical examination of coordination intensity, finding evidence that higher coordination intensity occupations are more resistant to displacement by AI based on historical growth in automation and OpenAI data on LLM exposure.

This paper contributes a model of human collaboration that we hope will be, in the words of a past reviewer, \textit{``a very elegant and fertile way to think differently from the `tasks being cut and parallelized' approach [to human computation] that is way more connected with how humans operate''}. 
It also contributes to a timely and important debate about the effects of AI on the economy and how recent advances may complement, rather than displace, skilled workers. 
Although recent research by OpenAI points out that nearly 20\% of the workforce may have over half of their tasks replaced by generative AI \citep{openai:laborimpact}, there may be ways it can enhance labor productivity \citep{Brynjolfsson:generative}. 
Some have suggested that understanding the impact of AI on labor will require connections between microlevel models of work and macrolevel models of markets~\cite{Frank2019-tg}. While our paper only provides a simple illustration of connecting a microlevel model of interdependencies with macrolevel insights on coordination intensity and wages, our empirical results provide further evidence that jobs requiring greater coordination intensity may be augmented by AI, or at least not replaced by it, and further motivate the need for models that emphasize collaborative capacities of human workers, bridge models of crowd work and traditional work, and promote AI in roles augmenting human collaboration.

\section{Modeling Complex Collaborative Work}\label{sec:modeling}

\subsection{Task graphs and workers}

Our model generalizes the assumption of small, independent micro-tasks in several ways. Like existing models, we represent a task as a collection of subtasks $V$. However, unlike most models, our subtasks can be dependent on other prerequisite subtasks through edges $E$, with weights $d_{uv}$, $d_v$ codifying levels of interdependencies between and within subtasks, as will be detailed later. Additionally, subtasks are not necessarily micro-tasks that can be completed by a single worker, but have a size parameter $s_v$ that can be arbitrarily large. We use $n = \lvert V \rvert$ to denote the number of subtasks and $V(G)$, $E(G)$ to denote the set of nodes and edges of the task graph $G$, and $N(v) := \{u \in V: (u,v) \in E\}$ to denote the prerequisite subtasks that $v$ is dependent on. Formally,

\begin{definition}\label{def:task-graphs}
A task graph is a weighted directed acyclic graph $G = (V, E, s, d)$ with sizes $s_v \in \mathcal{R}_+$ and interdependencies $d_v, d_{u, v} \in [0,1]$ for $v \in V$ and $(u,v) \in E$. 
\end{definition}
Individuals are represented as a set $I$. Each individual $i \in I$ has an \textit{available time} $t_i > 0$ representing the time they have and \textit{expertise} $e_{iv} > 0$ representing how time spent on subtask $v$ translates to work accomplished as a multiplicative factor\footnote{A multiplicative factor for scaling the time allocated towards work activities that translate into output is common in labor economics; see, for example, \cite{BilsCho1994:cyclic}.}. 
The baseline value of $e_{iv} = 1$ denotes that $1$ unit of time results in $1$ unit of work. If worker $i$ spends $t$ time on subtask $v$, they complete $e_{iv}t$ units of work. To complete subtask $v$ by themselves, they would need $s_v/e_{iv}$ time.

\subsection{Interdependencies and work}

The time it takes to complete a task is affected by interdependencies which capture the challenge of dividing work between many people due to lost context~\citep{Miller2016-im}. \textit{``Imagine the problem of writing code or writing an article, where we wish to enable a crowd to contribute iteratively.
A worker would have to know enough about what it is they should work on [and] how the subgoal may fit within the overall aim. But if the cost of understanding the context dominates the time that the worker is willing to contribute, [this] becomes costly, ineffective, or impossible''}~\cite{Zhang2011-cj}. 

Interdependencies $d_{uv}, d_v \in [0, 1]$ represent the amount of time workers need to spend picking up context as a multiplicative factor on prior work\footnote{The choice of $d_{uv}$, $d_v$ as multiplicative factors is purely notational to make them lie between $0$ and $1$ regardless of task size. One can recover non-multiplicative interdependencies, e.g. with an interdependency of $d_{uv} = c/s_{u}$ to obtain a constant cost $c$.}. If $u$ is a prerequisite of $v$, a worker $i$ assigned to $v$ would need to spend $s_ud_{uv}/e_{iu}$ time on context for $u$. If prior workers have already completed $\gamma$ units of work on $v$, then a worker $i$ assigned to continue work on $v$ would need to spend $\gamma d_v/e_{iv}$ time picking up this prior context. 
Varying $d_{uv}, d_v$ from $0$ to $1$ allows one to capture the entire range of possible dynamics — from work that can be done completely independently without penalty ($d = 0$), to work that is so interdependent that it is impossible to gain any benefit from dividing it up ($d = 1$) since the context time imposed is the same as the time required to redo prior work. These extreme cases correspond to typical assumptions of crowd work.

\subsection{Worker assignments and feasibility}

Large subtasks represent work that has not been decomposed, and thus may require multiple workers to complete. When multiple workers are assigned to a single subtask, they can either decompose the subtask themselves into smaller subtasks (out of scope\footnote{One can formalize workers discovering underlying structure by modeling partial knowledge of task structure as partitions of subtasks in a ``true'' task graph that gives rise to a ``known'' task graph.}), 
or work on the subtask sequentially.
\footnote{This does \textit{not} mean the model only captures sequential work. A task graph can have arbitrary numbers of subtasks that are carried out in parallel or in other complex configurations. It is only when multiple workers are assigned to a single subtask that they need to either carry out the work sequentially or have some other process for further dividing up the work. This setting generalizes the common one where subtasks are only assigned to individuals.}
An assignment describes the workers assigned to each subtask, and if there are multiple workers assigned to a given subtask, the order in which they carry out work. Formally,

\begin{definition}
An assignment $\mathcal{A}$ of workers to a task graph $G$ is a set of tuples $\{(\sigma_v, r_v)\}_{v \in V(G)}$, where $\sigma_v(k)$ denotes the $k$-th worker assigned to $v$ and $r_v(k)$ denotes the time that worker spends on $v$. We use $r_{iv}$ to denote the time allocated for worker $i$ to $v$, i.e. $r_{\sigma_v(k)v} = r_v(k)$.
\end{definition}

We note that $r_v(k)$ and $r_{iv}$ are defined to only include the time workers spend on work for $v$ or on context for prior work done on $v$. For notational convenience, they \textit{do not} include the time spent picking up context for \textit{prerequisite} subtasks of $v$, i.e. $\sum_{u \in N(v)}\frac{s_ud_{uv}}{e_{\sigma_v(k)u}}$. This needs to be added to obtain the total time a worker spends on their assignment.

Consider an assignment $(\sigma_v, r_v)$ of workers to a single subtask. The first individual assigned to $v$ is $\sigma_v(1)$, who spends time picking up context for subtasks in $N(v)$ and then up to $r_v(1)$ \textit{additional} time working on $v$. If this is not sufficient to complete $v$, then $\sigma_v(2)$ picks up where $\sigma_v(1)$ left off, spending time on context for subtasks in $N(v)$ and up to $r_v(2)$ time working on $v$. Note that $r_v(2)$ now also includes the time $\sigma_v(2)$ needs to spend on context for the prior work on $v$ done by $\sigma_v(1)$, but does \textit{not} include time spent on context for prerequisite subtasks of $v$. This repeats until either $v$ is completed or $\sigma_v$ is exhausted. If $v$ is completed, then $(\sigma_v, r_v)$ is \textit{$v$-feasible}. If all subtasks are completed, and if the total time spent by each worker $i$ is at most $t_i$, then the assignment $\mathcal{A}$ is \textit{feasible}. 

\begin{example}\label{ex:single-subtask}

Consider a subtask $v$ with size $s_v = 10$ and no prerequisite subtasks. Consider an assignment $\sigma_v = (i, j, k)$, $r_v = (5, 5, 5)$ of workers with baseline expertise $e_{iv} = e_{jv} = e_{kv} = 1$. Then $d_v = 0, (3-\sqrt{5})/2, 1$ represent scenarios from zero to complete interdependence:

{\bf No interdependencies, $d_v = 0$.} Worker $i$ spends all 5 units of time to complete 5 units of work. Since $d_v = 0$, $j$ does not spend any time picking up prior context and spends 5 units of time to complete the 5 remaining units of work. In this case, $r_v = (5, 5, 0)$ would have also been feasible.

{\bf Some interdependencies, $d_v = \frac{3-\sqrt{5}}{2}$.} After $i$ completes 5 units of work, $j$ spends $5(.3819...)$ units of time to pick up prior context, and then spends the remainder of their time to complete $3.09...$ additional units of work. Finally, $k$ spends $(5+3.09...)(0.3819...)$ time to pick up prior context, leaving just enough time to finish the remaining $1.91...$ units of work.

{\bf Full interdependencies, $d_v = 1$}. After $i$ completes 5 units of work, $j$ spends $5(1)$ time to pick up prior context, which uses up all their available time so that they cannot make any new contributions. The same holds for worker $k$. In this case, even assigning an infinite number of workers with the same time and expertise would not have been feasible.
\end{example}

Lemma \ref{lem:contribution} shows how the parts of the model work together to capture the contribution that a worker makes to a subtask, found by subtracting time spent on context from allocated time, and multiplying what remains by expertise.

\begin{lemma}\label{lem:contribution}
Given an assignment $\mathcal{A}$, let $c_v(k)$ denote the amount of time the $k$-th worker assigned to subtask $v$ spends on context for prior work done on $v$ (not including context for prerequisite subtasks). Let $w_v(k)$ denote the new contributions the $k$-th worker makes towards completing subtask $v$ and let $s_v(k) = \sum_{i=1}^k w_v(i)$ denote the total work completed on subtask $v$ by the first $k$ workers. Then,
\begin{align*}
    c_v(k) &= \max\left(\frac{s_v(k-1)d_v}{e_{\sigma_v(k)v}}, r_v(k)\right),\\
    w_v(k) &= \min\left(e_{\sigma_v(k)v}(r_v(k)-c_v(k)), s_v - s_v(k-1)\right).
\end{align*}
\end{lemma}
Suppose workers are only assigned to a single subtask. Then the assignment is $v$-feasible if and only if $s_v(\lvert \sigma_v \rvert) = s_v$ and $ t_{\sigma_v(k)} \geq r_v(k) + \sum_{u \in N(v)}\frac{s_ud_{uv}}{e_{\sigma_v(k)u}}$ for all $k$. 

\subsection{Empirical interpretation of interdependencies}\label{sec:empirical-interp}

Interdependencies can be in principle computed as the extra time required when a task is split between individuals.

\begin{example}\label{ex:empirical-interp}
Consider homogeneous workers $I$ with the same availability $t = 10$ and expertise $e$. Consider a single node task $v$ whose size $s_v$ and interdependency parameter $d_v$ are unknown. Suppose there are multiple instances of this task to complete. Now suppose that an individual finishes the first of these tasks using all $10$ units of their time, telling us that $et = s_v$. Now consider an individual who completes half of the second task in $5$ units of time before passing the remaining half to another worker who takes $8$ units of time to finish. Since workers are homogeneous, this extra time corresponds to the time the latter worker needs to pick up context that the first already has. Then $d_v$ can be calculated as $(8-5)/5 = 0.6$.
\end{example}

\section{Scalability of Complex Collaborative Work}\label{sec:scalability}

What makes collaborative work hard to scale? Why has independent, granular decomposition been critical for crowdsourcing, and how might we understand advances in moving past it? Our model provides simple expressions on the limits of feasible work that provide us with formal insights.

\subsection{Modeling work capacity}

One intuitively expects that highly modular tasks 
would be easier to scale compared to highly interdependent ones. To reason about the maximum size task that a set of workers can complete \textit{in relation to patterns of interdependence}, we need to formalize a family of task graphs with different size parameters but the same interdependency structure.
For task graph $G$, we use $h(G)$ to denote the family of task graphs that have the same interdependency structure as $G$ (the same set of edges, interdependency parameters, and relative task sizes), but that vary in total task size. Concretely,

\begin{definition}
For a task graph $G = (V, E, s, d)$, its interdependency family $h(G)$ consists of all task graphs $G' = (V', E', s', d')$  isomorphic under a one-to-one mapping $f: V' \to V$ satisfying
\begin{enumerate}
    \item $d'_v = d_{f(v)}, d'_{uv} = d_{f(u)f(v)}$ for all $v \in V', (u,v) \in E'$
    \item $s'_{v} = c s_{f(v)}$ for all $v \in V'$ and some constant $c$,
\end{enumerate}
\end{definition}

\begin{definition}
For a family of task graphs $\mathcal{G}$ and workers $I$, the work capacity $F_{I, \mathcal{G}}$ is the maximum size of tasks in $\mathcal{G}$ that can be completed by $I$. In other words, 
\begin{align*}
F_{I, \mathcal{G}} = \underset{G \in \mathcal{G}}{\text{max}} \left\{\sum_{v \in V(G)} s_v : \text{ a feasible $\mathcal{A}$ exists for $G$ } \right\}
\end{align*}
\end{definition}
This definition allows us to use $F_{I, h(G)}$ to represent the maximum size task that workers $I$ can complete for a given interdependency structure. One can then reason about scalability of that interdependency structure by analyzing how $F_{I_n, h(G)}$ grows given a \textit{sequence} of workers $I_1, I_2, \ldots, I_n$.

\subsection{Limits of feasible work even with infinite workers}

To what extent would \textit{infinite workers} enable the completion of complex goals? Theorem \ref{thm:infinite-workers-general} shows that for \textit{any non-zero interdependency} ($d > 0$), there is a fundamental limit to feasible task size.

\begin{theorem}\label{thm:infinite-workers-general}
Consider a task graph $G = (V, E, s, d)$. Let $I_n$ denote a set of $n$ homogeneous workers, i.e. $t_i = t$ and $e_{iv} = e_v$ for all $i \in I_n$ and $v \in V$. Let $s = \sum_{v \in V}s_v$. Then,
\begin{align}
    \lim_{n \to \infty} F_{I_n, h(G)} &= \frac{ts}{\max_{v \in V}\left[\frac{s_vd_{v}}{e_{v}} + \sum_{u \in N(v)}\frac{s_ud_{uv}}{e_{u}}\right]}.\label{eqn:infinite-workers-general}
\end{align}
For $d_{uv} = d_v = d$ and $e_v = e$, this simplifies to
\begin{align}
    \lim_{n \to \infty} F_{I_n, h(G)} 
    &= \frac{et\gamma_G}{d},\label{eqn:infinite-workers-capacity}
\end{align}
where $\gamma_G = s/\max_v \left[s_v + \sum_{u \in N(v)}s_u\right]$.
\end{theorem}
\begin{proof}
    We sketch the intuition here, leaving the formal $\epsilon-N$ proof for the Appendix. Since we have infinite homogeneous workers, the work capacity of the overall task can be found by considering infinite workers for \textit{each} subtask. 
    Lemma \ref{lem:v-feasible} gives us the maximum size for feasibility of each subtask, from which we can derive bounds for the overall task. 
\end{proof}

\begin{lemma}\label{lem:v-feasible}
Consider a task graph $G = (V, E, s, d)$ and an assignment $\mathcal{A} = \{(\sigma_v, r_v)\}_{v \in V}$. Suppose that the individuals assigned to $v$ each make some positive contribution, i.e. $w_v(k) > 0$.
Then $\mathcal{A}$ is $v$-feasible if and only if,
\begin{align}
\sum_{k=1}^{\lvert \sigma_v\rvert} e_{\sigma_v(k)v}r_v(k)\left(1-d_v\right)^{\lvert\sigma_v\rvert-k} \geq s_v.\label{eqn:subtask-feasibility}
\end{align}
\end{lemma}
\begin{proof}
This can be proved by induction, see Appendix.
\end{proof}

The parameter $\gamma_{G}$ is a more nuanced measure of task granularity. In the traditional view of task granularity, one wants to decompose tasks into subtasks that are as small as possible. A naive definition based on this view might represent task granularity as $\gamma_{G} = \frac{s}{\max_v s_v}$, with a larger $\gamma_{G}$ representing tasks that have been successfully decomposed into subtasks that are small relative to the overall task. 
Theorem \ref{thm:infinite-workers-general} reveals that it is $\gamma_G = s / \max_v \left[s_v + \sum_{u \in N(v)}s_u\right]$ that matters. In other words, definitions of task granularity need to consider prerequisite context costs too. 

\subsection{Why complex collaborative work is hard to scale} 
Theorem \ref{thm:infinite-workers-general} helps to reveal key factors limiting scalability of crowd work. High interdependencies (large $d$) and low task granularity (small $\gamma_G$) bound the work capacity to a constant factor of the contributions of top workers, which is in turn limited when workers are short-term novices (low $e$ and $t$). In \citep{Cebrian2016-oz}, the authors ask, \textit{``Given all we have learned about social mobilization, why isn’t social media a more reliable channel for constructive social change? ... [it] has been much better at providing the fuel for unpredictable, bursty mobilization than at steady, thoughtful construction of sustainable social change.''}
Equation \ref{eqn:infinite-workers-capacity} would note that constructive change requires highly interdependent work, but that social-media based mobilization typically results in large numbers of mostly short-term novices, fundamentally limiting the scope of feasible work.

\subsubsection{Trends in complex crowd work}

This also provides a simple explanation of trends in complex crowd work. Efforts in augmenting complex crowd work initially focused on designing micro-task workflows consisting of small subtasks with negligible interdependencies~\citep{Little2010-dx} so that any person could easily contribute. When these efforts reached a limit, it was followed by a shift from micro-tasks to larger macro-tasks and longer-term experts~\citep{Valentine2017-fs}. Each of these are efforts to nudge the key factors in Equation \ref{eqn:infinite-workers-capacity} towards a higher work capacity.




\subsubsection{Maximum context cost as the key bottleneck}\label{sec:key-bottleneck}

Equation \ref{eqn:infinite-workers-general} suggests that these factors (interdependencies, task granularity, time, expertise) boil down to a single quantity: the \textit{maximum time a worker may have to spend on context when assigned to a single subtask}, i.e. $c_v = \frac{s_vd_v}{e_v} + \sum_{u \in N(v)}\frac{s_vd_{uv}}{e_u}$. This is the true bottleneck limiting scalability, 
suggesting that a worthwhile direction is to find ways to augment people's ability to pass and receive context.

\subsection*{A sharp threshold around maximum context cost}

It turns out that $c_v$ also determines feasibility for a \textit{finite} number of homogeneous workers. Once $t$ grows just a little beyond $c_v$ to roughly $c_v(1 + \epsilon$), the number of workers required for $v$ quickly drops from infinity to less than $\frac{1}{d_v}\ln\frac{1}{\epsilon}$. 
\begin{theorem}\label{thm:sharp-threshold}
Consider a task graph $G = (V, E, s, d)$ and a set of homogeneous workers, each assigned to a single subtask. Let $n_v$ denote the number of workers assigned to $v$.
Then if $t \geq c_{v} + \frac{s_vd_v}{e_{v}}\cdot \frac{\epsilon}{1-\epsilon}$, for any $\epsilon > 0$, the assignment is feasible for $n_v \geq \frac{1}{d_v}\ln\frac{1}{\epsilon}$.
\end{theorem}
\begin{proof}
Using Lemma \ref{lem:v-feasible}, we get the maximum feasible size for $v$ to be $e_vt\sum_{k=1}^{n_v}(1-d_v)^{n_v - k} = e_vt\frac{1 - (1-d_v)^{n_v +1}}{d_v}$. Our result follows from noting that $(1-d_v)^{n_v} \leq \epsilon$ when $n_v \geq \ln\frac{1}{\epsilon}/\ln\frac{1}{1-d_v}$ and from noting that $\ln\frac{1}{\epsilon}/\ln\frac{1}{1-d_v} \leq \frac{1}{d_v}\ln\frac{1}{\epsilon}$ since $\ln(1+x) \geq x/(1+x)$ for $x > -1$.
\end{proof}

\section{Recruiting and Situated Learning}\label{sec:top-workers-learning}

\subsection{Top workers approximate work capacity}\label{sec:top-workers}

The finite worker case reveals another insight: when interdependencies exist, a small set of top workers maximizing \textit{expertise-weighted time} $e_{iv}r_{iv} = e_{iv}\left(t_i - \sum_{u \in N(v)}\frac{s_ud_{uv}}{e_{iu}}\right)$ essentially determine the work possible. All others can be dropped with only $\epsilon$ loss.

\begin{theorem}\label{thm:simple-top-workers}
For a task graph $G = (V, E, s, d)$, let $d_H$ denote the harmonic mean of $d_v$. Then for workers $I$ and any $\epsilon > 0$, there exists workers $I(\epsilon) \subseteq I$ such that, 
\begin{align}
    \lvert I(\epsilon) \rvert \leq \frac{\lvert V\rvert}{d_H}\ln\frac{1}{\epsilon},\text{  and  }
    F_{I(\epsilon),h(G)} \geq (1-\epsilon)F_{I, h(G)}.
\end{align}
\end{theorem}
\begin{proof}
The proof is by direct construction of $I(\epsilon)$. Let $I_v$ denote the workers assigned to $v$ in an optimal assignment and let $I_v(\epsilon)$ denote the "top" $\frac{1}{d_v}\ln\frac{1}{\epsilon}$ workers in $I_v$, selected to maximize expertise-weighted time $e_{iv}r_{iv}$. This totals $\frac{\lvert V\rvert}{d_H}\ln\frac{1}{\epsilon}$ workers across all $v$. We then use Lemma \ref{lem:v-feasible} to get the largest feasible $s_v$, as we detail in the Appendix.
\end{proof}


\subsubsection{Implications for upskilling}

This has important implications for learning in complex collaborative work. Theorem \ref{thm:simple-top-workers} says that for complex work where interdependencies are high (large $d_H$), there are few reasons for people to involve more than a small number of the top workers (small $\lvert I(\epsilon) \rvert$). 
This may be why platforms either center on independent micro-tasks that don't require any expertise or complex interdependent tasks for expert freelancers. There is no middle ground and no career development pathways for crowd workers to get upskilled from micro-tasks to freelance work~\cite{Rivera2021}. 
But how can people build skills in complex work if there are limited opportunities to engage while still a novice? Building expertise in these domains typically requires absorbing tacit knowledge that can only be learned experientially in the context of real-world work.

\subsection{Modeling legitimate peripheral participation}

One starting point is to consider existing methods for situated learning such as legitimate peripheral participation (LPP)~\citep{Lave1991-ds}, where novices move from peripheral to core tasks by picking up context while making simple contributions. While LPP has been shown to support novice learning, especially in the context of open-source, its limitations are also well documented. Novices face numerous participation barriers~\citep{Steinmacher2014-vq} and most core work is still concentrated to a small number of long-term experts~\citep{Maass2004-lo}. Is it possible to formally reason about the challenges of legitimate peripheral participation? 

\subsubsection{Task ecosystems and expertise} To do this, we overload our use of task graphs to represent not just a single task instance, but an entire \textit{uniform task ecosystem} with many task instances that all share a common set of subtasks and interdependencies represented by a shared task graph $(V, E, s, d)$. 
For example, in an ecosystem centered on ``developing front-end web applications using Angular'', subtask types might include ``defining Angular routes'' or ``implementing the view of a component''. 

Worker expertise depends on the subtask type. In other words, if a worker is proficient at ``implementing the view of a component'', that proficiency applies to all such tasks within the ecosystem. For our purposes, we let each node $v$ in the shared task graph represent a different subtask type with $e_{iv}$ denoting the expertise of worker $i$ for task type $v$.

\subsubsection{Learning and LPP} 

In our model of LPP, a novice $i$ begins with expertise $e_{\text{novice}} = e_{\text{normal}}/M$ for all task types, where $M$ is the multiplicative gap between a typical worker and a complete novice. If a worker completes subtask type $v$ as the sole assigned worker, then their expertise for that task type increases to $e_{\text{normal}}$ (they have upskilled in that task type). Importantly, however, a worker is only assigned to a subtask if they can complete it within some required time $\tau$.\footnote{Even in volunteer contexts, there are many reasons why tasks not completed within some time frame might never be completed, e.g. loss of motivation, increased likelihood of errors, etc.}

\begin{definition}\label{def:lpp-accessible}
A subtask type $v$ is \textit{trivially accessible} if any novice can directly complete it, i.e. $\frac{s_v}{e_{\text{novice}}} + \sum_{u \in N(v)} \frac{s_ud_{uv}}{e_{\text{novice}}} \leq \tau$.
A subtask type $v$ is \textit{LPP-accessible} if it is trivially accessible or if a novice can eventually be assigned to complete it after completing other trivially accessible or LPP-accessible subtask types. Formally, $v$ is LPP-accessible if a set of LPP-accessible subtask types $P = \{v_1, v_2, \ldots, v_k\}$ exists such that,
\begin{align*}
    \frac{s_v}{e_{\text{novice}}} + \sum_{u \in N(v) \setminus P} \frac{s_ud_{uv}}{e_{\text{novice}}} + \sum_{u \in N(v)\cap P} \frac{s_ud_{uv}}{e_{\text{normal}}}\leq \tau.
\end{align*}
\end{definition}

\subsection*{LPP requires appropriate interdependencies}

With these definitions, we can now formally discuss what is required for a task ecosystem to successfully support novice upskilling through legitimate peripheral participation.
\begin{theorem}\label{thm:lpp-accessible}
Consider a task ecosystem with shared task graph $(V, E, s, d)$. For subtask type $v$, let $T_v$ denote the time cost of completing subtask type $v$ given fully upskilled workers, i.e. $e_{iv} = e_{iu} = e_{\text{normal}}$ for all $u \in N(v)$. Let $P \subseteq N(v)$ denote the prerequisite subtask types of $v$ that are LPP-accessible, and let $\alpha_v$ denote the fraction of $T_v$ that fully upskilled workers would need to spend on context costs for subtask types in $P$. Then $v$ is trivially accessible if and only if $T_v \leq \frac{\tau}{M}$ and $v$ is \textit{LPP-accessible} if and only if,
    \begin{align*}
        \alpha_v &\geq \frac{MT_v - \tau}{MT_v - T_v}
    \end{align*}
\end{theorem}
\begin{proof}
     Since $P$ is defined as subtask types that are LPP-accessible, Definition \ref{def:lpp-accessible} tells us that $v$ is LPP-accessible so long as $\frac{s_v}{e_{\text{novice}}} + \sum_{u \in N(v) \setminus P} \frac{s_ud_{uv}}{e_{\text{novice}}} + \sum_{u \in N(v)\cap P} \frac{s_ud_{uv}}{e_{\text{normal}}}\leq \tau$. Careful work shows this to be equivalent to $MT_v - (M-1)\alpha_vT_v \leq \tau$, which gives us our result (see Appendix).
\end{proof}

Given that large tasks and significant interdependencies are obstacles for novices, it is not surprising that breaking down tasks so that $T_v \leq \frac{\tau}{M}$ will make them accessible to novices. However, Theorem \ref{thm:lpp-accessible} also points out that if these small tasks do not help bridge the gap towards more complex ones, they don't help beyond providing expertise in commoditized tasks that anyone could do.
For LPP to create \textit{pathways} from novice to expert tasks, one actually needs sufficient interdependence ($\alpha_v \geq \frac{MT_v - \tau}{MT_v - T_v}$). There needs to be subtask types that, despite having a large total cost, have much of that cost tied up in context costs that can be reduced through LPP-based upskilling.
LPP has been likened to moving through layers of an onion representing progressive sets of a task~\cite{Yunwen_Ye2003-ir}. Our model would say that defining such layers is not just about progressing in difficulty, but about particular patterns of interdependence. One needs ecosystems where involvement in early layers reduces context for later ones. 

\section{Coordination Intensity, Wages, and AI}\label{sec:coordination-intensity}

We now turn to traditional work as a context where complex collaborative work is the norm. We have already seen that interdependencies increase the value of expertise. Assuming all else is equal, we would expect firms to put more money into hiring high-expertise workers for tasks with higher interdependence than for those where expertise is not as critical, and that this would be reflected in wages.

\subsection{Modeling occupations and coordination intensity}

To explore this intuition, we model the coordination intensity of occupations, derive its relationship with hourly wages in a simple setting, and show that our theoretical analyses are consistent with empirical analyses of occupational employment and wage data in the U.S. economy.

\subsubsection{Occupations, expertise, and wages}\label{sec:model-coordination}

We represent occupation $O \subseteq V$ as a set of subtasks assigned to the same worker with $s_O = \sum_{v \in O}s_v$ denoting the size of occupation $O$. 
\begin{definition}
For occupation $O$, we say that a worker $i$ has occupational expertise $e_{iO}$ for occupation $O$ if $e_{iv} = e_{iO}$ for all $v \in O$. 
We assume workers are assigned to one occupation and have baseline expertise for all other subtasks. 

\end{definition}
\begin{definition}\label{def:market-wage}
For occupation $O$, the wage function $w_O : \mathbb{R}^+ \to \mathbb{R}^+$ maps occupational expertise to wages (per unit time) for a worker at that expertise level, and is of the form
\begin{align*}
    w_O(x) = x^d \quad\textit{for $d > 1$}.
\end{align*}
\end{definition}
The constraint $d > 1$ follows from the fact that a worker with expertise $ce$ can provide at least the same value as $c$ workers of expertise $e$, so $w_O(c\cdot e) \geq c\cdot w_O(e)$ for $c \geq 1$.

\subsubsection*{Coordination intensity}

We model coordination intensity $C_{iO}$ as the time worker $i$ spends on context due to interdependencies between subtasks of their own occupation and subtasks assigned to others. It is these interdependencies that manifest themselves as observable coordination costs that would be reflected in empirical work activity data. Formally,

\begin{definition}\label{def:coordination-intensity}
For occupation $O$ and worker $i$, the coordination intensity is
\begin{align*}
    C_{iO} = \sum_{u \in V}\max_{v: v \in O, u \in N(v)} \frac{s_u d_{uv}}{e_{iu}}.
\end{align*}
Because workers are assumed to have baseline expertise for subtasks outside their occupation, this reduces coordination intensity to the following worker-independent expression\footnote{We note that the $\max$ in this expression comes from making sure we are not double-counting context costs. The expression sums over all subtasks for which we need to pay context costs (the prerequisite subtasks of subtasks in the occupation). 
For each such subtask $u$, the required context cost depends on the $d_{uv}$ that imposes the highest context cost (max $d_{uv}$ for $v \in O$).},
\begin{align*}
    C_O = \sum_{u \in V}\max_{v: v \in O, u \in N(v)} s_u d_{uv}.
\end{align*}
\end{definition}

\subsection{Theoretical analysis of coord. intensity and wages}

We can now discuss a firm's decision to recruit workers at differing expertise levels for a set of required occupations needed to make a given project deadline $\tau$. 

\begin{definition}
We represent a simple dynamic for occupational recruitment with the tuple $(\mathcal{O}, \tau)$, where $\mathcal{O}$ represents the set of occupations that a firm is recruiting workers for and $\tau$ representing the time remaining till the project deadline at which recruited workers need to complete their work. We use $\rho_O = \frac{s_O}{\tau - C_O}$ to denote the minimum expertise required for a worker to complete $O$ by the project deadline.
\end{definition}
Given a set of occupations $\mathcal{O}$, we use $C^* = \max_{O \in \mathcal{O}} C_O$ to denote the occupation with highest coordination intensity, $O^*$ to denote the corresponding occupation, and $\rho_{O^*}$ to denote the expertise required for a worker to finish the maximum coordination intensity occupation by the deadline.


\begin{theorem}\label{thm:coord-wages}
    Consider occupational recruitment with parameters $(\mathcal{O}, \tau)$ and let $J = \max_{O \in \mathcal{O}}\frac{s_{O^*}}{s_O}$ denote the multiplicative gap between the largest and smallest occupational sizes. Then if $\alpha = \frac{C_O}{C^*} \leq \frac{1}{2}$ and $C^* \geq \frac{2}{3}\tau$, we have
    \begin{align*}
         \alpha + \ln\left(\frac{1}{2J}\frac{s_{O^*}}{C^*}\right) \leq \ln(\rho_O) \leq 2\alpha + \ln\left(\frac{s_{O^*}}{C^*}\right).
    \end{align*}
\end{theorem}
\begin{proof}
    To get our approximation for $\alpha \leq \frac{1}{2}$, we first show that $\rho_O = \frac{1}{2}\frac{s_O}{s_{O^*}}\text{HM}\left(\rho_{O^*}, \frac{s_{O^*}}{C^* - C_O}\right)$, where HM denotes the harmonic mean. We then note that the harmonic mean satisfies $\min(a,b) \leq HM(a,b) \leq 2\min(a,b)$; that $\frac{J}{C_{\max} - C(O)} \leq e_{\max}$ when $C_{\max} \geq \frac{2}{3}\tau$; and that $\frac{x}{1+x} \leq \ln(1+x) \leq x$ for $x > 0$ (see Appendix for details).
\end{proof}

Since $s_{O^*}, C^*, J$ are constants, Theorem \ref{thm:coord-wages} shows a linear relationship between $\alpha = C_O/C^*$ and $\ln(\rho_O)$. For a market wage function satisfying Definition \ref{def:market-wage}, i.e. $w_O(x) = x^d$ for $d > 1$, it follows that $\ln(w_O) = d\ln(\rho_O)$, giving a linear relationship between coordination intensity $C_O$ and log wages $\ln(w_O)$. As we will see, this fits our empirical analysis well. 

\begin{figure}[h] \centering 
	\includegraphics[width=0.7\linewidth]{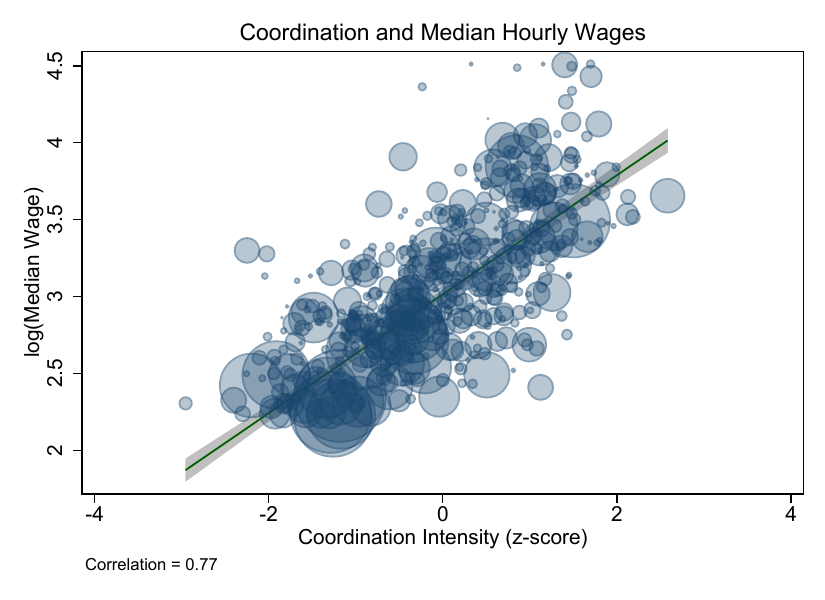}
	\caption{\small 
 Sources: O*NET (2014-2018) and the Occupational Employment Statistics (2010-2018). The plot depicts the relationship between the intensity of work activities requiring coordination for an occupation with logged median hourly wages deflated using the 2012 PCE index. 
    Each circle represents one occupation, with size representing occupational employment. 
    To construct a comprehensive measure of coordination intensity, we take the unweighted average of the following work activity sub-indices: ``Getting Information,'' ``Monitor Processes, Materials, or Surroundings,'' ``Processing Information,'' ``Communicating with Supervisors, Peers, or Subordinates,'' ``Organizing, Planning, and Prioritizing Work,'' and ``Coordinating the Work and Activities of Others.'' We have experimented with alternative formulations of our coordination index and our results are robust to the exclusion or addition of certain terms. Observations are weighted by occupational employment through matching the data on coordination intensities with the Occupational Employment and Wage Statistics (OEWS) at an occupation level of aggregation. The OEWS provides nationally representative employment counts and hourly wages at a six-digit standard occupational classification (SOC) level since 2005.
} \label{fig:scatter_coordination_loccemp}
\end{figure}

\begin{figure*}[t]
\centering
\includegraphics[width=\textwidth]{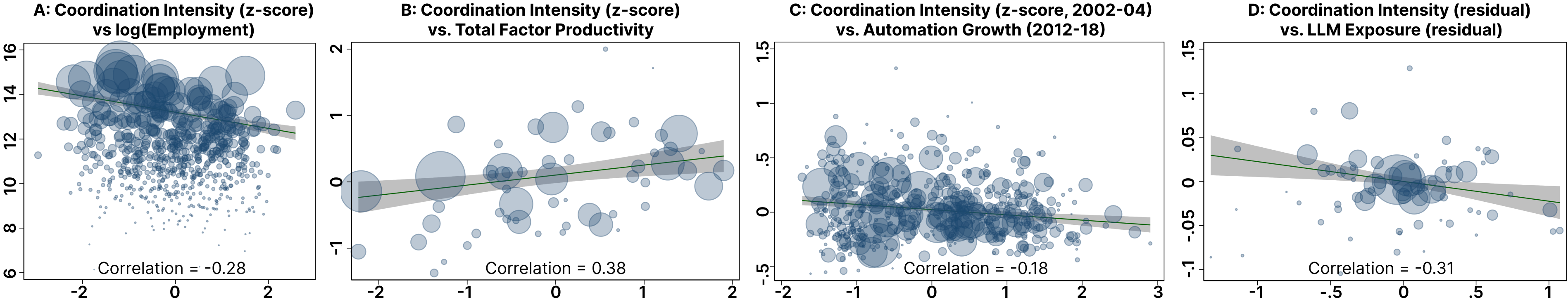} 
\caption{\small \textbf{Panel A} plots the relationship between coordination intensity and employment at a six-digit occupational level from 2014-2018. \textbf{Panel B} plots the relationship between employment-weighted coordination intensity and total factor productivity (TFP) across 46 industries. We construct TFP by taking the residual from a regression of logged real GDP in 2012 prices on logged capital expenditures and logged employment. To map our index of occupational coordination intensity into a three-digit NAICS code, we create weighted averages using employment at the industry $\times$ occupation level. \textbf{Panel C} plots the relationship between coordination intensity and the growth in automation as a work characteristic between 2002 and 2018 at a six-digit occupational level. \textbf{Panel D} plots the relationship between coordination intensity and exposure to LLMs across 71 three-digit NAIC industries, residualizing both variables on demographic controls and skill intensities that are contained in Tables S1-S4 in the SA. We use the OESW data at the industry and occupational level to obtain employment to construct an employment-weighted industry index of coordination intensity at the three-digit NAICS level to match the OpenAI index.}
\label{fig:coordination-graphs}
\end{figure*}

\subsection{Empirical analysis of coord. intensity and wages}\label{sec:coordination-and-wages-empirical}

We used the O*NET database to conduct an empirical analysis of coordination intensity and wages.
O*NET is a survey that the U.S. Department of Labor administers to a random sample of U.S. workers within occupations. Respondents answer questions on an ordinal scale that generally measures both the importance of a task and the frequency at which different tasks occur on the job. We focus on work activities, taking the product of the importance and frequency weights (when available) to generate an overall intensity for each sub-index task \citep{AutorHandel2013}. We gather every wave of the survey between 2002 and 2018 to produce a harmonized time series. More data details in figure caption.

We found a strong positive correlation $(\rho = 0.77)$ between the intensity of coordination-related work activities and logged median wages (\textbf{Figure \ref{fig:scatter_coordination_loccemp}}), showing that occupations with higher coordination intensity tend to have a higher wage. A standard deviation (1sd) increase in coordination intensity is associated with a 46\% rise in hourly wages. To demonstrate this correlation is not spurious, Table 1 in the Appendix presents the results in a multivariate setting with controls over occupational composition and skill intensities.


\subsection{Further empirical results on coordination intensity}

Some have suggested that understanding the impact of AI on labor will require connections between microlevel models of work and macrolevel models of markets~\cite{Frank2019-tg}. While our analysis only lightly illustrated this (connecting a microlevel model of interdependencies with macrolevel insights on coordination intensity and wages), it highlights the opportunity of models that emphasize collaborative capacities of workers. To future motivate this, we conclude with additional empirical analyses relating coordination intensity to employment, productivity, and exposure to AI. We use the same measure of occupational coordination intensity. Further data details are included in the figure caption.

\subsubsection*{Coordination Intensity, Employment, and Productivity}

In Figure \ref{fig:coordination-graphs}, \textbf{Panel A} shows a negative correlation with employment $(\rho = -0.28)$. A 1sd increase in coordination intensity is associated with a 0.36\% decline in employment. This is consistent with intuition that higher levels of coordination costs would lead one to hire a smaller number of top experts. \textbf{Panel B} shows a positive correlation with total factor productivity (TFP) ($\rho = 0.37$). A 1sd rise in coordination intensity is associated with a 0.15\% rise in TFP. This is consistent with intuition that tasks with higher coordination costs would be harder to disaggregate and commoditize, so would have higher productivity.
In unreported results, we find an even stronger correlation with value added per worker, a proxy for labor productivity ($\rho = 0.54$). We find that a 1sd rise in coordination intensity is associated with a 0.15\% rise in TFP and a 0.32\% rise in value added per worker\footnote{The results are robust to alternative weighting schemes. For example, an unweighted average across occupations in the same industry produce correlations of 0.16 and 0.47 between coordination intensity and TFP and value added per worker. Employment-weighted averages produce correlations of 0.39 and 0.52, respectively.}. 
Tables 1-2 in the Appendix replicate these results in a multivariate context controlling for demographic characteristics and skill intensities.


\subsubsection*{Coordination Intensity and Exposure to AI}

In Figure \ref{fig:coordination-graphs}, \textbf{Panel C} plots the relationship between coordination intensity in 2002-2004 and the growth in automation from 2002 to 2018, showing a negative correlation ($\rho = -0.18$). A 1sd rise in coordination intensity is associated with a 4 percentage point (pp) drop in growth in automation---economically meaningful given that the average occupation grows by 3.2pp in automation intensity. This is consistent with intuition that occupations with higher coordination intensity would be less exposed to AI because they would be harder to decompose and automate. Table 3 in the Appendix replicates these results in a multivariate context controlling for demographic characteristics and skill intensities.

Next, we explore the relationship between coordination intensity and large language model (LLM) exposure. However, the raw correlation would be misleading because many higher skilled occupations are more exposed to LLMs~\citep{openai:laborimpact}. After controlling for demographic characteristics and skill intensities across industries, \textbf{Panel D} plots the relationship between \textit{residualized} coordination intensity and LLM exposure, producing a negative correlation ($\rho = -0.31$). We used OpenAI's data on LLM exposure~\cite{openai:laborimpact} created from detailed human coding of occupational tasks from O*NET to determine the share of an occupations tasks for which LLM technologies could reduce the time it takes to complete the task by at least 50\%. While they have not published individual occupational exposures, they did publish employment-weighted industry-level exposures at the three-digit NAICS level, which we use. Table 4 in the Appendix begins with the raw (positive) correlation, showing that sequentially adding controls to mitigate bias flips this to a negative correlation. These results suggest that jobs with greater coordination intensity are less, not more, exposed to LLMs. 

\section{Conclusion}\label{sec:conclusion}

In this paper, we argued for models of crowd work that capture the unique human capacity for collaboration. 
We introduced a graph-based model building on the concepts of task complexity and interdependencies to capture the challenge of dividing up work. We illustrated how such a model can provide insights on the limits of social-media based mobilization, trends in complex crowd work, the outsized importance of expert workers, the role of interdependencies in legitimate peripheral participation, and the relationship between coordination intensity and wages.

Our model makes it possible to ask a number of questions that we view as promising directions. What generative models best represent the structure of complex work? 
Can collaborative structures or new AI technologies help reduce the cost of context? How might we create work ecosystems that support continuous on the job learning? To reach a human-centered vision of Human-AI collaboration and to reach a \textit{``future crowd workplace in which we would want our children to participate''}~\cite{Kittur2013-qr}, we will need models that can help us to reason about uniquely human capacities for collaborative work, that bridge models of crowd work and traditional work, and that make it possible to ask how algorithms can play an augmenting role in these dynamics.

%
%
%
%
%

\bibliographystyle{ACM-Reference-Format}
\bibliography{human-computation}

\newpage

\appendix
\section{Proofs for Scalability of Complex Collaborative Work}

\subsection{Proof of Theorem \ref{thm:infinite-workers-general}}

Theorem \ref{thm:infinite-workers-general} shows that for any non-zero interdependency, there is a fundamental limit to feasible task size even with infinite workers. To prove it, we start by deriving the maximum size that a node $v$ can be for $(\sigma_v, r_v)$ to be $v$-feasible. This is then applied to each node of a task graph to derive the overall work capacity.

\begin{lemma}\label{lem:v-feasible}
Consider a task graph $G = (V, E, s, d)$ and an assignment $\mathcal{A} = \{(\sigma_v, r_v)\}_{v \in V}$. Suppose that the individuals assigned to $v$ each make some positive (non-zero) contribution to the subtask, i.e. $w_v(k) > 0$.
Then $\mathcal{A}$ is $v$-feasible if and only if,
\begin{align*}
    \sum_{k=1}^{\lvert \sigma_v\rvert} e_{\sigma_v(k)v}r_v(k)\left(1-d_v\right)^{\lvert\sigma_v\rvert-k} \geq s_v.
\end{align*}
\end{lemma}
\begin{proof}
For $k < \lvert \sigma_v\rvert$, we give a proof by induction that $s_v(k)$, the amount of work completed by the first $k$ individuals, can be expressed as
\begin{align*}
    s_v(k) = \sum_{i=1}^k e_{\sigma_v(i)v}r_v(i)\left(1-d_v\right)^{k-i}.
\end{align*}
For the base case, we know that $s_v(0) = 0$. Then assuming the statement holds for $k-1$, we have,
\begin{align*}
    s_{v}(k) &= s_v(k-1) + w_v(k) \text{ (by definition of $s_v(k)$)}\\
    &= s_v(k-1) + e_{\sigma_v(k)v}(r_v(k)-c_v(k)) \text{ (by definition of $w_v(k)$, $k < \lvert \sigma_v\rvert$, $w_v(k+1) > 0$)}\\
    &= s_v(k-1) + e_{\sigma_v(k)v}\left[r_v(k) - \frac{d_vs_{v}(k-1)}{e_{\sigma_v(k)v}}\right]\text{ (by definition of $c_v(k)$, $w_v(k) > 0$)}\\
    &= \left(1-d_v\right)s_{v}(k-1) +     e_{\sigma_v(k)v}r_v(k)\text{ (rearranging terms)}\\
    &= \sum_{i=1}^k e_{\sigma_v(i)v}r_v(i)\left(1-d_v\right)^{k-i}\text{ (induction hypothesis)}
\end{align*}
This proves the expression for $s_v(k)$, $k < \lvert \sigma_v\rvert$. We can then use this to derive an expression for $s_v(k)$ when $k = \lvert \sigma_v \rvert$.
\begin{align*}
    s_{v}(k) &= s_v(k-1) + w_v(k) \text{ (by definition of $s_v(k)$)}\\
    &= s_v(k-1) + \min(e_{\sigma_v(k)v}(r_v(k)-c_v(k)), s_v - s_v(k-1)) \text{ (by definition of $w_v(k)$)}\\
    &= \min(s_v(k-1) + e_{\sigma_v(k)v}(r_v(k)-c_v(k)), s_v)\text{ (rearranging terms)}\\
    &= \min\left(\sum_{i=1}^k e_{\sigma_v(i)v}r_v(i)\left(1-d_v\right)^{k-i}, s_v\right) \text{ (using the same logic as for $k < \lvert \sigma_v \rvert$)}
\end{align*}
Therefore,
\begin{align*}
    s_v(\lvert \sigma_v\rvert) &= \min\left(\sum_{i=1}^{\lvert\sigma_v\rvert} e_{\sigma_v(i)v}r_v(i)\left(1-d_v\right)^{\lvert\sigma_v\rvert-i}, s_v\right)
\end{align*}
By definition, $\mathcal{A}$ is $v$-feasible if and only if $s_v(\lvert \sigma_v \rvert) \geq s_v$, which concludes our proof.
\end{proof}

The informal intuition for proving Theorem \ref{thm:infinite-workers-general} is a straightforward application of Lemma \ref{lem:v-feasible}. Since we have an infinite number of homogeneous workers, the maximum feasible size of the overall task graph can be found by considering an infinite number of homogeneous workers for \textit{each node} in the task graph. (We don't need to worry about who is assigned to what node since workers are homogeneous). Lemma \ref{lem:v-feasible} gives us the maximum size for each node, which in turn gives us bounds for the overall task since each node is a fixed fraction of the overall task. We give the more formal (but less intuitive) $\epsilon\text{-}N$ proof below.

\begin{proof}
To prove the limit, it is sufficient to show that for any $\epsilon > 0$, there exists some $N(\epsilon)$ such that for all $n \geq N(\epsilon)$, 
\begin{align*}
    \frac{ts}{\max_{v \in V}\left[\frac{s_vd_{v}}{e_{v}} + \sum_{u \in N(v)}\frac{s_ud_{uv}}{e_{u}}\right]} \geq F_{I_{n}, h(G)} \geq \frac{ts}{\max_{v \in V}\left[\frac{s_vd_{v}}{e_{v}(1-\epsilon)} + \sum_{u \in N(v)}\frac{s_ud_{uv}}{e_{u}}\right]}
\end{align*}
Let $N(\epsilon) = \sum_v n_v(\epsilon)$, where $n_v(\epsilon) = \ln\frac{1}{\epsilon}/\ln\frac{1}{1-d_v}$. Consider $n \geq n(\epsilon)$ workers who are assigned only to subtask $v$ with all their time allocated to that subtask. Then from Lemma \ref{lem:v-feasible}, we know that this assignment is $v$-feasible if and only if
\begin{align*}
    s_v &\leq \sum_{i=1}^{n} e_{v}\left[t - \sum_{u \in N(v)}\frac{s_ud_{uv}}{e_{u}}\right]\left(1-d_v\right)^{n-i}\\
    &= e_{v}\left[t - \sum_{u \in N(v)}\frac{s_ud_{uv}}{e_{u}}\right]\frac{1-(1-d_v)^{n}}{d_v}
\end{align*}
Since $n \geq n_v(\epsilon)$, the assignment is $v$-feasible if (but not only if)
\begin{align*}
    s_v &\leq \frac{e_{v}(1-\epsilon)}{d_v}\left[t - \sum_{u \in N(v)}\frac{s_ud_{uv}}{e_{u}}\right],
\end{align*}
which is true if and only if
\begin{align*}
    t &\geq \frac{s_vd_v}{e_v(1-\epsilon)} + \sum_{u \in N(v)}\frac{s_ud_{uv}}{e_{u}} \iff s \leq \frac{ts}{\frac{s_vd_{v}}{e_{v}(1-\epsilon)} + \sum_{u \in N(v)}\frac{s_ud_{uv}}{e_{u}}}
\end{align*}
Since this holds for all $v \in V$, we get the right hand side of our desired equation
\begin{align*}
F_{I_{n}, h(G)} \geq \frac{ts}{\max_{v \in V}\left[\frac{s_vd_{v}}{e_{v}(1-\epsilon)} + \sum_{u \in N(v)}\frac{s_ud_{uv}}{e_{u}}\right]}
\end{align*}
To get the left-hand side, we note that 
\begin{align*}
    s_v &< e_{v}\left[t - \sum_{u \in N(v)}\frac{s_ud_{uv}}{e_{u}}\right]\frac{1}{d_v}.
\end{align*}
which means that the assignment is NOT feasible at $s_v = e_{v}\left[t - \sum_{u \in N(v)}\frac{s_ud_{uv}}{e_{u}}\right]\frac{1}{d_v}$. Following the same logic then gives us the left-hand side
\begin{align*}
F_{I_{n}, h(G)} \leq \frac{ts}{\max_{v \in V}\left[\frac{s_vd_{v}}{e_{v}} + \sum_{u \in N(v)}\frac{s_ud_{uv}}{e_{u}}\right]}
\end{align*}
Finally, we can simply plug in $d_{uv} = d$, $d_v = d$, and $e_v = e$ to get the simplified form.
\end{proof}

\subsection*{Proof for Theorem \ref{thm:sharp-threshold}}

This theorem follows fairly directly from the geometric nature of the work capacity expressed in Lemma \ref{lem:v-feasible}.

\begin{proof}
From Lemma \ref{lem:v-feasible}, we know that for each $v$, the assignment is $v$-feasible if $e_vt\sum_{k=1}^{n_v}(1-d_v)^{n_v - k} = e_vt\frac{1 - (1-d_v)^{n_v +1}}{d_v} \geq s_v$. We note that since $\ln(1+x) \geq x/(1+x)$ for $x > -1$, we have
\begin{align*}
    \ln\frac{1}{\epsilon}/\ln\frac{1}{1-d_v} = \ln\frac{1}{\epsilon}/\ln(1+\frac{d_v}{1-d_v}) \leq \frac{1}{d_v}\ln\frac{1}{\epsilon}
\end{align*}
We also note that when $n_v \geq \ln\frac{1}{\epsilon}/\ln\frac{1}{1-d_v}$, we have 
\begin{align*}
    (1-d_v)^{n_v} \leq (1-d_v)^{\ln\frac{1}{\epsilon}/\ln\frac{1}{1-d_v}} = (1-d_v)^{\ln\epsilon/\ln(1-d_v)} = e^{\ln(1-d_v)\ln\epsilon/\ln(1-d_v)} = \epsilon
\end{align*}
Putting these together, this means that when $n_v \geq \frac{1}{d_v}\ln\frac{1}{\epsilon}$, it must be true that $(1-d_v)^{n_v} \leq \epsilon$. If we also have $t \geq c_{v} + \frac{s_vd_v}{e_{v}}\cdot \frac{\epsilon}{1-\epsilon}$, then
\begin{align*}
    e_vt\frac{1 - (1-d_v)^{n_v +1}}{d_v} \geq e_vt\frac{1-\epsilon(1-d_v)}{d_v} \geq e_v\left(c_{v} + \frac{s_vd_v}{e_{v}}\cdot \frac{\epsilon}{1-\epsilon}\right)\frac{1-\epsilon(1-d_v)}{d_v}\\
    \geq e_v\left(\frac{s_vd_v}{e_v} + \frac{s_vd_v}{e_v}\cdot \frac{\epsilon}{1-\epsilon}\right)\frac{1-\epsilon(1-d_v)}{d_v} = e_v\left(\frac{s_vd_v}{e_v}\cdot \frac{1}{1-\epsilon}\right)\frac{1-\epsilon(1-d_v)}{d_v} \geq s_v
\end{align*}
Since this is true for all $v$, then the assignment must be feasible.
\end{proof}

\section*{Proofs for Recruiting and Situated Learning}

\subsection*{Proof for Theorem \ref{thm:simple-top-workers}}

The intuition of Theorem \ref{thm:simple-top-workers} is that one can simply pick the top $\frac{1}{d_v}\ln\frac{1}{\epsilon}$ workers for each node of an optimal assignment. Lemma \ref{lem:v-feasible} can then be used to show that the contributions from the remaining participants is negligible for every node.

\begin{theorem}\label{thm:simple-top-workers}
For a task graph $G = (V, E, s, d)$, let $d_H$ denote the harmonic mean of $d_v$. Then for workers $I$ and any $\epsilon > 0$, there exists workers $I(\epsilon) \subseteq I$ such that,
\begin{align*}
    \lvert I(\epsilon) \rvert \leq \frac{\lvert V\rvert}{d_H}\ln\frac{1}{\epsilon},\text{  and  }
    F_{I(\epsilon),h(G)} \geq (1-\epsilon)F_{I, h(G)}.
\end{align*}
\end{theorem}
\begin{proof}
Let $\mathcal{A} = \{(\sigma_v, r_v)\}_{v \in V}$ be an optimal assignment that achieves the work capacity $F_{I, h(G)}$. We assume that this assignment orders assigned individuals in an ascending manner by expertise-weighted time $e_{iv}\left[t_i - \sum_{u \in N(v)}\frac{s_ud_{uv}}{e_{iu}}\right]$, for any node $v \in V$. We can assume this because given any optimal assignment, a rearrangement of individuals in this way can only increase the capacity (as can be seen by the expression in Lemma \ref{lem:v-feasible}) and so must also be optimal. 

Let $I_v = \{ i \in I : \sigma_v(k) = i \text{ for some $k$ }\}$ denote the workers assigned to $v$. For each $v \in V$, let $I_v(\epsilon)$ denote the top $\frac{1}{d_v}\ln\frac{1}{\epsilon}$ workers in $I_v$ for which $e_{iv}\left[t_i - \sum_{u \in N(v)}\frac{s_ud_{uv}}{e_{iu}}\right]$ is maximized, and let $I(\epsilon) = \cup_{v \in V}I_v(\epsilon)$. Then,
\begin{align*}
    \lvert I(\epsilon) \rvert &= \sum_{v \in V}\frac{1}{d_v}\ln\frac{1}{\epsilon} = \frac{\lvert V\rvert}{d_H}\ln\frac{1}{\epsilon}.
\end{align*}
By Lemma \ref{lem:v-feasible}, we know that the maximum capacity of individuals $I_v$ assigned to a single node $v$ is
\begin{align*}
    s_v^* = \sum_{i=1}^{\lvert \sigma_v \rvert} e_{\sigma_v(i)v}r_v(i)\left(1-d_v\right)^{\lvert\sigma_v\rvert-i}
\end{align*}
Let $n_{v1} = \frac{1}{d_v}\ln\frac{1}{\epsilon}$ represent the number of top workers and $n_{v2} = \lvert \sigma_v \rvert - n_{v1}$ be the number of remaining workers. Since individuals are assigned in ascending order and $I_v(\epsilon)$ contains the top workers, we know that the maximum capacity of individuals $I_v(\epsilon)$ assigned to a single node $v$ is
\begin{align*}
    s_v^{top} = \sum_{i=n_{v2}+1}^{\lvert\sigma_v\rvert} e_{\sigma_v(i)v}r_v(i)\left(1-d_v\right)^{\lvert\sigma_v\rvert-i} = s_v^* - \sum_{i=1}^{n_{v2}} e_{\sigma_v(i)v}r_v(i)\left(1-d_v\right)^{\lvert\sigma_v\rvert-i},
\end{align*}
We now show that $s_v^{top} \geq (1-\epsilon)s_v^*$. To show $s_v^{top} \geq (1-\epsilon)s_v^*$, we can show that $s_v^* - s_v^{top} \leq \epsilon s_v^*$, as follows,
\begin{align*}
    s_v^* - s_v^{top} &= \sum_{i=1}^{n_{v2}} e_{\sigma_v(i)v}r_v(i)\left(1-d_v\right)^{\lvert\sigma_v\rvert-i}\\
    &= (1-d_v)^{n_{v1}}\sum_{i=1}^{n_{v2}} e_{\sigma_v(i)v}r_v(i)\left(1-d_v\right)^{\lvert\sigma_v\rvert-i-n_{v1}}\text{ (bringing out common factors)}\\
    &\leq \epsilon\sum_{i=1}^{n_{v2}} e_{\sigma_v(i)v}r_v(i)\left(1-d_v\right)^{\lvert\sigma_v\rvert-i-n_{v1}}\text{ (by $(1-d_v)^{n_{v1}}\leq \epsilon$)}\\
    &\leq \epsilon\sum_{i=n_{v1} +1}^{\lvert\sigma_v \rvert} e_{\sigma_v(i-n_{v1})v}r_v(i-n_{v1})\left(1-d_v\right)^{\lvert \sigma_v\rvert -i}\text{ (by change of variables)}\\
    &\leq \epsilon\sum_{i=n_{v1} +1}^{\lvert\sigma_v \rvert} e_{\sigma_v(i)v}r_v(i)\left(1-d_v\right)^{\lvert \sigma_v\rvert -i}\text{ (by ascending order of $e_{\sigma_v(i)v}r_v(i)$)}\\
    &\leq \epsilon s_v^*.
\end{align*}
Since this holds for all nodes, this gives us our desired result. 
\end{proof}

\subsection*{Proof for Theorem \ref{thm:lpp-accessible}}

The intuition is that if $v$ is to be LPP-accessible, then there must be a sufficient amount of the context costs that it can reduce through upskilling in other tasks first through the LPP process.

\begin{theorem}\label{thm:lpp-accessible}
Consider a task ecosystem with shared task graph $(V, E, s, d)$. For subtask type $v$, let $T_v$ denote the time cost of completing subtask type $v$ given fully upskilled workers, i.e. $e_{iv} = e_{iu} = e_{\text{normal}}$ for all $u \in N(v)$. Let $P \subseteq N(v)$ denote the prerequisite subtask types of $v$ that are LPP-accessible, and let $\alpha_v$ denote the fraction of $T_v$ that fully upskilled workers would need to spend on context costs for subtask types in $P$. Then $v$ is trivially accessible if and only if $T_v \leq \frac{\tau}{M}$ and $v$ is \textit{LPP-accessible} if and only if,
    \begin{align*}
        \alpha_v &\geq \frac{MT_v - \tau}{MT_v - T_v}
    \end{align*}
\end{theorem}
\begin{proof}
We know by definition that $T_v$ is defined as the time that fully upskilled workers would take to finish $v$, i.e. $T_v = \frac{s_v}{e_{\text{normal}}} + \sum_{u \in N(v)} \frac{s_ud_{uv}}{e_{\text{normal}}}$. We also know that $v$ is trivially accessible if and only if a novice can be directly assigned to complete it within the required time, i.e. $\frac{s_v}{e_{\text{novice}}} + \sum_{u \in N(v)} \frac{s_ud_{uv}}{e_{\text{novice}}} \leq \tau$. Then the fact that $e_{\text{novice}} = e_{\text{normal}}/M$ directly gives us our result that $v$ is trivially accessible if and only if $T_v \leq \frac{\tau}{M}$.

We know by definition that $\alpha_v$ is the fraction of $T_v$ that fully upskilled workers need to spend on context costs for subtask types in $P$, i.e. $\alpha_v = \frac{\sum_{u \in N(v)\cap P} \frac{s_ud_{uv}}{e_{\text{normal}}}}{T_v}$. We also know that by definition, $v$ is LPP-accessible if and only if some set of LPP-accessible tasks $Q$ exists such that
\begin{align*}
    \frac{s_v}{e_{\text{novice}}} + \sum_{u \in N(v) \setminus Q} \frac{s_ud_{uv}}{e_{\text{novice}}} + \sum_{u \in N(v)\cap Q} \frac{s_ud_{uv}}{e_{\text{normal}}}\leq \tau
\end{align*}
Let us first prove that if $\alpha_v \geq \frac{MT_v - \tau}{MT_v - T_v}$, then $v$ is LPP accessible. Choose $Q = P$. Then we have
\begin{align*}
    \frac{s_v}{e_{\text{novice}}} + \sum_{u \in N(v) \setminus P} \frac{s_ud_{uv}}{e_{\text{novice}}} + \sum_{u \in N(v)\cap P} \frac{s_ud_{uv}}{e_{\text{normal}}} = MT_v - (M-1)\sum_{u \in N(v)\cap P} \frac{s_ud_{uv}}{e_{\text{normal}}}\\
    = MT_v - (M-1)\alpha_vT_v \leq MT_v - (M-1)\left(\frac{MT_v - \tau}{MT_v - T_v}\right)T_v = \tau
\end{align*}
which gives us our result. Now let us prove that if $\alpha_v < \frac{MT_v - \tau}{MT_v - T_v}$, then $v$ is not LPP accessible. We will prove this by contradiction. Suppose that $\alpha_v < \frac{MT_v - \tau}{MT_v - T_v}$ but $v$ is LPP accessible. 
Then this means there must exist some $Q$ that satisfies $\frac{s_v}{e_{\text{novice}}} + \sum_{u \in N(v) \setminus Q} \frac{s_ud_{uv}}{e_{\text{novice}}} + \sum_{u \in N(v)\cap Q} \frac{s_ud_{uv}}{e_{\text{normal}}}\leq \tau$. Since $P$ is defined as the set of \textit{all} prerequisite subtask types of $v$ that are LPP accessible, it must be true that $Q \subseteq P$. Let $R = P \setminus Q$. Then we have,
\begin{align*}
    \frac{s_v}{e_{\text{novice}}} + \sum_{u \in N(v) \setminus Q} \frac{s_ud_{uv}}{e_{\text{novice}}} + \sum_{u \in N(v)\cap Q} \frac{s_ud_{uv}}{e_{\text{normal}}} &= 
    \frac{s_v}{e_{\text{novice}}} + \sum_{u \in N(v) \setminus P} \frac{s_ud_{uv}}{e_{\text{novice}}} + \sum_{u \in N(v)\cap P} \frac{s_ud_{uv}}{e_{\text{normal}}} + (M-1)\sum_{u \in N(v)\cap R} \frac{s_ud_{uv}}{e_{\text{normal}}}\\
    &= MT_v - (M-1)\alpha_vT_v + (M-1)\sum_{u \in N(v)\cap R} \frac{s_ud_{uv}}{e_{\text{normal}}}\\
    &> MT_v - (M-1)\left(\frac{MT_v - \tau}{MT_v - T_v}\right)T_v + (M-1)\sum_{u \in N(v)\cap R} \frac{s_ud_{uv}}{e_{\text{normal}}}\\
    &= \tau + (M-1)\sum_{u \in N(v)\cap R} \frac{s_ud_{uv}}{e_{\text{normal}}}\\
    &\geq \tau
\end{align*}
But this is a contradiction since we showed that the left-hand side is both $\geq \tau$ and $< \tau$, so such a $Q$ must not exist.
\end{proof}

\section*{Proofs for Coordination Intensity and Wages}

\subsection*{Proof for Theorem \ref{thm:coord-wages}}

\begin{lemma}\label{lem:expertise-hm}
Consider occupational recruitment with parameters $(\mathcal{O}, \tau)$. Then letting $HM$ denote the harmonic mean,
\begin{align*}
    \rho_O = \frac{1}{2}\frac{s_O}{s_{O^*}}\text{HM}\left(\rho_{O^*}, \frac{s_{O^*}}{C^* - C_O}\right),
\end{align*}
\end{lemma}
\begin{proof}
    We know that $\rho_{O}(\tau-C_O) = s_O$ and that $\rho_{O^*}(\tau-C^*) = s_{O^*}$. This means that $\tau = \frac{s_{O^*}}{\rho_{O^*}} + C^*$. Plugging this into the former equation and solving for $\rho_{O}$, we get 
    $\rho_{O} = \frac{s_O}{\frac{s_{O^*}}{\rho_{O^*}} + (C^* - C_O)} = \frac{1}{2}\frac{s_O}{s_{O^*}}\text{HM}\left(\rho_{O^*}, \frac{s_{O^*}}{C^* - C_O}\right)$.
\end{proof}

\begin{theorem}\label{thm:coord-wages}
    Consider occupational recruitment with parameters $(\mathcal{O}, \tau)$ and let $J = \max_{O \in \mathcal{O}}\frac{s_{O^*}}{s_O}$ denote the multiplicative gap between the largest and smallest occupational sizes. Then if $\alpha = \frac{C_O}{C^*} \leq \frac{1}{2}$ and $C^* \geq \frac{2}{3}\tau$, we have
    \begin{align*}
         \alpha + \ln\left(\frac{1}{2J}\frac{s_{O^*}}{C^*}\right) \leq \ln(\rho_O) \leq 2\alpha + \ln\left(\frac{s_{O^*}}{C^*}\right).\label{eqn:coord-approx}
    \end{align*}
\end{theorem}
\begin{proof}
    First we note that if $\alpha = \frac{C_O}{C^*} \leq \frac{1}{2}$ and $C^* \geq \frac{2}{3}\tau$, then we have that
    \begin{align*}
        \frac{s_{O*}}{C^* - C_O} \leq \frac{s_{O*}}{\frac{1}{2}C^*} \leq \frac{s_{O*}}{\frac{1}{3}\tau} \leq \frac{s_{O*}}{\tau - C^*} = \rho_{O^*}
    \end{align*}
    This means that we can now use Lemma \ref{lem:expertise-hm} and the fact that the harmonic mean satisfies $\min(a,b) \leq HM(a,b) \leq 2\min(a,b)$ to get the following
    \begin{align*}
        \frac{1}{2}\frac{s_{O}}{C^* - C_O} \leq \rho_O \leq \frac{s_{O}}{C^* - C_O}
    \end{align*}
    Then with careful approximations using the fact that $\frac{x}{1+x} \leq \ln(1+x) \leq x$ for $x > 0$, we have
    \begin{align*}
        \ln(\rho_O) &\leq \ln\left(\frac{s_{O}}{C^* - C_O}\right)\\
        &= \ln\left(\frac{s_{O}}{C^*}\right) + \ln\left(\frac{1}{1 - \frac{C_O}{C^*}}\right)
        = \ln\left(\frac{s_{O}}{C^*}\right) + \ln\left(1 + \frac{\frac{C_O}{C^*}}{1 - \frac{C_O}{C^*}}\right)\\
        &= \ln\left(\frac{s_{O}}{C^*}\right) + \frac{\frac{C_O}{C^*}}{1 - \frac{C_O}{C^*}}
        = \ln\left(\frac{s_{O}}{C^*}\right) + 2\alpha \leq \ln\left(\frac{s_{O^*}}{C^*}\right) + 2\alpha
    \end{align*}
    Similarly, on the other side, we have
    \begin{align*}
        \ln(\rho_O) &\geq \ln\left(\frac{1}{2}\frac{s_{O}}{C^* - C_O}\right)\\
        &= \ln\left(\frac{1}{2}\frac{s_{O}}{C^*}\right) + \ln\left(\frac{1}{1 - \frac{C_O}{C^*}}\right)
        = \ln\left(\frac{1}{2}\frac{s_{O}}{C^*}\right) + \ln\left(1 + \frac{\frac{C_O}{C^*}}{1 - \frac{C_O}{C^*}}\right)\\
        &= \ln\left(\frac{1}{2}\frac{s_{O}}{C^*}\right) + \frac{\frac{C_O}{C^*}/(1 - \frac{C_O}{C^*})}{1+\frac{C_O}{C^*}/(1 - \frac{C_O}{C^*})}
        = \ln\left(\frac{1}{2}\frac{s_{O}}{C^*}\right) + \alpha \geq \ln\left(\frac{1}{2J}\frac{s_{O^*}}{C^*}\right) + \alpha 
    \end{align*}
\end{proof}

\section*{Additional Tables for Empirical Analyses of Coordination Intensity}

Table \ref{tab:occupation_coordinate} begins by replicating the main result that occupations with higher coordination intensity have higher hourly wages and lower employment. We control for demographic characteristics and skill intensities to mitigate concerns about omitted variables bias. In particular, we use the American Community Survey (ACS) five-year sample from 2021 to semi-parametrically control for the demographic distribution within every five-digit SOC code, including: the number of children, age, gender, race, education, and degree of rural versus urban residence. For example, occupations with more educated workers could have higher wages for reasons unrelated to coordination intensity. We also control for skill intensities, including: reading, active listening, critical thinking, problem solving, judgment, systems analysis, and operations analysis skills. These controls help ensure that we are not conflating coordination intensity with occupations that simply require other sorts of skills. Although the statistical significance declines as we add these controls from column 1 to 3, the economic significance and negative relationship remains. The positive correlation with hourly wages, however, remains both statistically and economically significant (columns 4-6).

\begin{table}[htbp]\centering
\def\sym#1{\ifmmode^{#1}\else\(^{#1}\)\fi} \caption{Controlling for Observables in Employment and Wage Regressions \label{tab:occupation_coordinate}}
\resizebox{11cm}{!}{
\begin{tabular}{l*{6}{c}} 
\hline\hline
\small
\input{tables/occupation_coordinate.tex}
\hline\hline 
\end{tabular}}
\fnote{Notes.---Sources: American Community Survey (5-year 2021) and O*NET. The table reports the coefficients associated with regressions of the log number of employees and the log median hourly wage (deflated using the 2012 personal consumption expenditure index) on the standardized z-score of coordination intensity at the occupational level, controlling for various demographic characteristics and skill intensities. We gather the demographic characteristics from the ACS as five-digit SOC occupational averages. We also gather various skill intensities to control for additional dimensions of unobserved heterogeneity. Standard errors are clustered at the occupational level and observations are weighted by employment.}
\end{table}

Table \ref{tab:industry_coordinate} duplicates these results for other outcomes of interest. In particular, we use the OEWS data to create employment-weighted industry measures of coordination intensity and correlate it with measures of industry productivity, such as real GDP and total factor productivity (TFP). We construct TFP by regressing log real GDP on employment and capital, taking the residual as TFP. Again, while the relationships are not statistically significant, particularly as we only have 46 industries available in the matched sample, the estimates are still economically meaningful and consistent with our earlier results on hourly wages. Table \ref{tab:occupation_dlautomate} also replicates the results in the main text linking growth in automation intensity between 2002 and 2018 with coordination intensity in a multivariate context, remaining economically meaningful.

\begin{table}[htbp]\centering
\def\sym#1{\ifmmode^{#1}\else\(^{#1}\)\fi} \caption{Controlling for Observables in Industry Productivity Regressions \label{tab:industry_coordinate}}
\resizebox{11cm}{!}{
\begin{tabular}{l*{6}{c}} 
\hline\hline
\input{tables/industry_coordinate.tex}
\hline\hline 
\end{tabular}}
\fnote{Notes.---Sources: American Community Survey (5-year 2021), Occupational Employment and Wage Statistics, Bureau of Economic Analysis, and O*NET. The table reports the coefficients associated with regressions of the logged real GDP and logged total factor productivity (TFP) on the standardized z-score of coordination intensity at the NAICS-3 industry level, controlling for various demographic characteristics and skill intensities that have been aggregated across all occupations $\times$ industry using the OEWS data and weighted by the wage bill (employment $\times$ earnings).  BEA are presented in real 2012 prices and TFP is calculated by regressing logged GDP in a sector on logged employment and capital, taking the residual. We gather the demographic characteristics from the ACS as five-digit SOC occupational averages. We also gather various skill intensities to control for additional dimensions of unobserved heterogeneity. Standard errors are clustered at the NAICS-3 level and observations are weighted by employment.}
\end{table}

\begin{table}[htbp]\centering
\def\sym#1{\ifmmode^{#1}\else\(^{#1}\)\fi} \caption{Controlling for Observables in Automation Growth Regressions \label{tab:occupation_dlautomate}}
\resizebox{7cm}{!}{
\begin{tabular}{l*{6}{c}} 
\hline\hline
\input{tables/occupation_dlautomate.tex}
\hline\hline 
\end{tabular}}
\fnote{Notes.---Sources: American Community Survey (5-year 2021) and O*NET. The table reports the coefficients associated with regressions of the growth rate of an index of automation of tasks between 2002-04 and 2014-18 averages on the standardized z-score of coordination intensity at the occupational level, controlling for various demographic characteristics and skill intensities (2014-18 occupational averages). Our index of automation is based on work context information, i.e. how easily tasks are automated within an occupation. We gather the demographic characteristics from the ACS as five-digit SOC occupational averages. We also gather various skill intensities to control for additional dimensions of unobserved heterogeneity. Standard errors are clustered at the occupational level and observations are weighted by employment.}
\end{table}

Finally, Table \ref{tab:chatgpt_coordinate} presents preliminary results linking exposure to LLMs from Eloundou et al. (2023) with employment-weighted coordination intensity across 71 industries. Columns 1 and 4 present the raw cross-sectional correlation between coordination intensity and the two measures of LLM exposure. Crucially, as we layer additional controls, even standard demographic characteristics in columns 2 and 5, the marginal effect turns negative. Furthermore, additional controls for industry-weighted skill requirements in columns 3 and 6, the coefficients remain negative. That only the estimate in column 3 is statistically significant at the 10\% level reflects the small sample size. In short, the negative relationship between coordination intensity and exposure to LLMs suggests that it has the potential to augment---or at least not displace---jobs that require greater degrees of coordination. A natural question is why omitted variables play a more substantial role in the context of LLM exposure than in the preceding tables. As Eloundou et al. (2023) point out, higher skilled occupations are more exposed to LLMs, so failing to control for confounding factors here will create bias in the opposite direction.

\begin{table}[htbp]\centering
\def\sym#1{\ifmmode^{#1}\else\(^{#1}\)\fi} \caption{Preliminary Evidence on Large Language Model Exposure and Coordination Intensity \label{tab:chatgpt_coordinate}}
\resizebox{10cm}{!}{
\begin{tabular}{l*{6}{c}} 
\hline\hline
\input{tables/chatgpt_coordinate.tex}
\hline\hline 
\end{tabular}}
\fnote{Notes.---Sources: American Community Survey (5-year 2021), Occupational Employment and Wage Statistics (OEWS), and O*NET. The table reports the coefficients associated with regressions of the standardized z-score of coordination intensity at the NAICS-3 industry level on a standardized score of exposure to large language model (LLM) from Eloundou et al. (2023), controlling for various demographic characteristics and skill intensities that have been aggregated across all occupations $\times$ industry using the OEWS data and weighted by the wage bill (employment $\times$ earnings). Eloundou et al. (2023) construct two measures of LLM exposure based on task data: one using human coders, and another using ChatGPT itself. The correlation between the two for the 71 industries is 0.97. We gather the demographic characteristics from the ACS as five-digit SOC occupational averages. We also gather various skill intensities to control for additional dimensions of unobserved heterogeneity. Standard errors are clustered at the NAICS-3 level and observations are weighted by employment.}
\end{table}
\end{document}

%% file: tables/occupation_coordinate.tex
                    &\multicolumn{3}{c}{Log Employment}                               &\multicolumn{3}{c}{Log Median Hourly Wage}                       \\\cmidrule(lr){2-4}\cmidrule(lr){5-7}
                    &\multicolumn{1}{c}{(1)}&\multicolumn{1}{c}{(2)}&\multicolumn{1}{c}{(3)}&\multicolumn{1}{c}{(4)}&\multicolumn{1}{c}{(5)}&\multicolumn{1}{c}{(6)}\\
Coordination intensity&       -.349\sym{***}&       -.144         &       -.157         &        .378\sym{***}&        .107\sym{***}&        .053\sym{**} \\
                    &      [.132]         &      [.109]         &      [.141]         &      [.029]         &      [.020]         &      [.022]         \\
Number of children  &                     &       1.173         &        .871         &                     &       -.280\sym{*}  &       -.151         \\
                    &                     &     [1.416]         &     [1.329]         &                     &      [.165]         &      [.169]         \\
Age 20-30, \%       &                     &      -5.334         &      -5.859         &                     &      -2.976\sym{***}&      -2.712\sym{***}\\
                    &                     &     [6.116]         &     [5.687]         &                     &      [.719]         &      [.683]         \\
Age 31-40, \%       &                     &      -8.307\sym{***}&      -8.629\sym{**} &                     &       -.515         &       -.845\sym{**} \\
                    &                     &     [3.205]         &     [3.411]         &                     &      [.429]         &      [.421]         \\
Age 41-50, \%       &                     &      -7.293         &      -7.597         &                     &      -1.607\sym{**} &      -1.240         \\
                    &                     &     [6.528]         &     [6.577]         &                     &      [.802]         &      [.799]         \\
Age 50+             &                     &      -2.827         &      -3.803         &                     &      -2.952\sym{***}&      -2.713\sym{***}\\
                    &                     &     [4.656]         &     [4.227]         &                     &      [.546]         &      [.477]         \\
Male, \%            &                     &       -.734\sym{*}  &       -.210         &                     &        .118\sym{***}&        .094         \\
                    &                     &      [.405]         &      [.436]         &                     &      [.045]         &      [.058]         \\
Married, \%         &                     &      -2.834         &      -1.156         &                     &       2.534\sym{***}&       1.941\sym{***}\\
                    &                     &     [2.700]         &     [2.805]         &                     &      [.434]         &      [.422]         \\
White, \%           &                     &      -1.188         &      -3.655\sym{**} &                     &       -.125         &        .021         \\
                    &                     &     [1.534]         &     [1.792]         &                     &      [.239]         &      [.249]         \\
Black, \%           &                     &       -.177         &      -1.772         &                     &        .105         &        .144         \\
                    &                     &     [2.212]         &     [2.443]         &                     &      [.389]         &      [.366]         \\
College degree, \%  &                     &       2.305\sym{**} &       1.953\sym{*}  &                     &        .396\sym{***}&        .228\sym{*}  \\
                    &                     &     [1.114]         &     [1.053]         &                     &      [.130]         &      [.127]         \\
Masters degree, \%  &                     &       -.138         &       -.470         &                     &        .403\sym{**} &        .186         \\
                    &                     &      [.981]         &      [.981]         &                     &      [.175]         &      [.171]         \\
PhD degree, \%      &                     &      -3.574\sym{*}  &      -3.829\sym{*}  &                     &        .627\sym{**} &        .325         \\
                    &                     &     [2.138]         &     [2.064]         &                     &      [.297]         &      [.338]         \\
Lives in metro, \%  &                     &      -5.293\sym{*}  &      -6.078\sym{**} &                     &       -.283         &       -.121         \\
                    &                     &     [2.726]         &     [2.754]         &                     &      [.323]         &      [.280]         \\
Reading skills      &                     &                     &       -.224         &                     &                     &        .154\sym{***}\\
                    &                     &                     &      [.318]         &                     &                     &      [.037]         \\
Active listening skills&                     &                     &        .728\sym{**} &                     &                     &       -.097\sym{***}\\
                    &                     &                     &      [.297]         &                     &                     &      [.029]         \\
Critical thinking skills&                     &                     &       -.441         &                     &                     &        .002         \\
                    &                     &                     &      [.319]         &                     &                     &      [.045]         \\
Problem solving skills&                     &                     &       -.340         &                     &                     &        .076\sym{*}  \\
                    &                     &                     &      [.314]         &                     &                     &      [.039]         \\
Judgment skills     &                     &                     &        .562\sym{*}  &                     &                     &        .023         \\
                    &                     &                     &      [.312]         &                     &                     &      [.042]         \\
Systems analysis skills&                     &                     &       -.163         &                     &                     &       -.018         \\
                    &                     &                     &      [.272]         &                     &                     &      [.036]         \\
Operations analysis skills&                     &                     &       -.016         &                     &                     &        .037\sym{*}  \\
                    &                     &                     &      [.144]         &                     &                     &      [.022]         \\
R-squared           &         .07         &         .26         &         .31         &         .59         &         .91         &         .92         \\
Sample Size         &         492         &         492         &         492         &         475         &         475         &         475         \\
Demographic Controls&          No         &          No         &          No         &          No         &          No         &          No         \\
Skill Controls      &          No         &         Yes         &         Yes         &          No         &         Yes         &         Yes         \\

%% file: tables/industry_coordinate.tex
                    &\multicolumn{3}{c}{Log Real GDP}                                 &\multicolumn{3}{c}{Log Total Factor Productivity}                \\\cmidrule(lr){2-4}\cmidrule(lr){5-7}
                    &\multicolumn{1}{c}{(1)}&\multicolumn{1}{c}{(2)}&\multicolumn{1}{c}{(3)}&\multicolumn{1}{c}{(4)}&\multicolumn{1}{c}{(5)}&\multicolumn{1}{c}{(6)}\\
Coordination intensity&        .113         &        .164         &       1.365         &        .115\sym{**} &       -.130         &        .128         \\
                    &      [.130]         &      [.334]         &      [.873]         &      [.049]         &      [.202]         &      [.441]         \\
Number of children  &                     &       6.772         &      17.402\sym{***}&                     &       1.899         &       8.344\sym{***}\\
                    &                     &     [4.226]         &     [5.047]         &                     &     [3.007]         &     [2.900]         \\
Age 20-30, \%       &                     &     -56.823\sym{*}  &     -16.452         &                     &     -26.261         &     -24.500         \\
                    &                     &    [30.968]         &    [33.964]         &                     &    [16.249]         &    [17.967]         \\
Age 31-40, \%       &                     &       8.550         &     -51.413\sym{*}  &                     &       -.865         &     -27.180\sym{**} \\
                    &                     &    [19.067]         &    [26.213]         &                     &    [13.021]         &    [11.656]         \\
Age 41-50, \%       &                     &     -50.961         &     -13.663         &                     &     -17.086         &     -23.614         \\
                    &                     &    [33.793]         &    [36.059]         &                     &    [16.057]         &    [20.897]         \\
Age 50+             &                     &     -21.700         &     -13.893         &                     &     -12.576         &     -16.784\sym{*}  \\
                    &                     &    [13.525]         &    [15.398]         &                     &     [8.670]         &     [8.639]         \\
Male, \%            &                     &       2.079         &       8.789\sym{***}&                     &       1.956         &       3.802\sym{***}\\
                    &                     &     [2.296]         &     [2.139]         &                     &     [1.374]         &     [1.276]         \\
Married, \%         &                     &     -24.767\sym{*}  &     -27.760\sym{**} &                     &      -8.779         &     -13.946\sym{*}  \\
                    &                     &    [14.728]         &    [11.597]         &                     &     [8.140]         &     [7.533]         \\
White, \%           &                     &       -.991         &       1.392         &                     &       -.300         &        .777         \\
                    &                     &     [4.319]         &     [4.890]         &                     &     [3.018]         &     [2.768]         \\
Black, \%           &                     &      -9.622\sym{*}  &     -16.982\sym{***}&                     &        .289         &      -3.450         \\
                    &                     &     [5.257]         &     [4.859]         &                     &     [3.602]         &     [2.645]         \\
College degree, \%  &                     &       4.466         &       9.801\sym{*}  &                     &       3.172         &       5.587\sym{*}  \\
                    &                     &     [3.179]         &     [5.565]         &                     &     [1.912]         &     [3.098]         \\
Masters degree, \%  &                     &       1.847         &      18.067         &                     &       5.398         &       9.467         \\
                    &                     &     [9.809]         &    [12.751]         &                     &     [6.506]         &     [8.429]         \\
PhD degree, \%      &                     &      21.155         &     -32.106         &                     &       7.817         &     -23.084         \\
                    &                     &    [17.781]         &    [35.927]         &                     &    [13.891]         &    [20.286]         \\
Lives in metro, \%  &                     &      -5.989         &      -2.562         &                     &      -4.883         &      -2.038         \\
                    &                     &     [8.264]         &     [5.727]         &                     &     [4.890]         &     [2.906]         \\
Reading skills      &                     &                     &       3.863\sym{***}&                     &                     &       2.377\sym{***}\\
                    &                     &                     &     [1.085]         &                     &                     &      [.544]         \\
Active listening skills&                     &                     &        .122         &                     &                     &      -1.344\sym{**} \\
                    &                     &                     &      [.994]         &                     &                     &      [.651]         \\
Critical thinking skills&                     &                     &      -2.297         &                     &                     &        .020         \\
                    &                     &                     &     [2.306]         &                     &                     &     [1.499]         \\
Problem solving skills&                     &                     &      -1.989         &                     &                     &       -.322         \\
                    &                     &                     &     [1.314]         &                     &                     &      [.804]         \\
Judgment skills     &                     &                     &        .157         &                     &                     &       -.246         \\
                    &                     &                     &     [1.912]         &                     &                     &     [1.011]         \\
Systems analysis skills&                     &                     &        .488         &                     &                     &        .399         \\
                    &                     &                     &     [1.153]         &                     &                     &      [.731]         \\
Operations analysis skills&                     &                     &      -1.273\sym{**} &                     &                     &       -.828\sym{***}\\
                    &                     &                     &      [.554]         &                     &                     &      [.306]         \\
R-squared           &         .03         &         .43         &         .71         &         .09         &         .33         &         .66         \\
Sample Size         &          46         &          46         &          46         &          46         &          46         &          46         \\
Demographic Controls&          No         &          No         &          No         &          No         &          No         &          No         \\
Skill Controls      &          No         &         Yes         &         Yes         &          No         &         Yes         &         Yes         \\

%% file: tables/occupation_dlautomate.tex
                    &\multicolumn{3}{c}{Automation Growth (2002-18)}                  \\\cmidrule(lr){2-4}
                    &\multicolumn{1}{c}{(1)}&\multicolumn{1}{c}{(2)}&\multicolumn{1}{c}{(3)}\\
Coordination intensity&       -.044\sym{*}  &       -.036         &       -.056         \\
                    &      [.024]         &      [.027]         &      [.035]         \\
Number of children  &                     &        .115         &        .100         \\
                    &                     &      [.312]         &      [.300]         \\
Age 20-30, \%       &                     &       -.330         &       -.311         \\
                    &                     &     [1.783]         &     [1.680]         \\
Age 31-40, \%       &                     &        .999         &       1.243         \\
                    &                     &      [.981]         &      [.856]         \\
Age 41-50, \%       &                     &       -.131         &       -.327         \\
                    &                     &     [2.206]         &     [2.032]         \\
Age 50+             &                     &       -.361         &       -.113         \\
                    &                     &     [1.307]         &     [1.211]         \\
Male, \%            &                     &        .032         &        .087         \\
                    &                     &      [.082]         &      [.104]         \\
Married, \%         &                     &       -.696         &       -.892         \\
                    &                     &      [.661]         &      [.673]         \\
White, \%           &                     &       -.306         &       -.431         \\
                    &                     &      [.375]         &      [.384]         \\
Black, \%           &                     &       -.142         &       -.525         \\
                    &                     &      [.597]         &      [.605]         \\
College degree, \%  &                     &        .327\sym{*}  &        .309         \\
                    &                     &      [.191]         &      [.197]         \\
Masters degree, \%  &                     &       -.100         &       -.138         \\
                    &                     &      [.152]         &      [.160]         \\
PhD degree, \%      &                     &        .694\sym{**} &        .643\sym{*}  \\
                    &                     &      [.352]         &      [.329]         \\
Lives in metro, \%  &                     &       -.091         &       -.105         \\
                    &                     &      [.581]         &      [.569]         \\
Reading skills      &                     &                     &        .024         \\
                    &                     &                     &      [.062]         \\
Active listening skills&                     &                     &        .084\sym{*}  \\
                    &                     &                     &      [.044]         \\
Critical thinking skills&                     &                     &       -.144\sym{**} \\
                    &                     &                     &      [.061]         \\
Problem solving skills&                     &                     &        .089         \\
                    &                     &                     &      [.060]         \\
Judgment skills     &                     &                     &        .000         \\
                    &                     &                     &      [.051]         \\
Systems analysis skills&                     &                     &        .023         \\
                    &                     &                     &      [.033]         \\
Operations analysis skills&                     &                     &       -.025         \\
                    &                     &                     &      [.035]         \\
R-squared           &         .03         &         .14         &         .17         \\
Sample Size         &         431         &         431         &         431         \\
Demographic Controls&          No         &          No         &          No         \\
Skill Controls      &          No         &         Yes         &         Yes         \\

%% file: tables/chatgpt_coordinate.tex
                    &\multicolumn{3}{c}{LLM Exposure (Human)}                         &\multicolumn{3}{c}{LLM Exposure (Computer)}                      \\\cmidrule(lr){2-4}\cmidrule(lr){5-7}
                    &\multicolumn{1}{c}{(1)}&\multicolumn{1}{c}{(2)}&\multicolumn{1}{c}{(3)}&\multicolumn{1}{c}{(4)}&\multicolumn{1}{c}{(5)}&\multicolumn{1}{c}{(6)}\\
Coordination intensity&        .022\sym{***}&       -.013\sym{*}  &       -.022\sym{*}  &        .028\sym{***}&       -.007         &       -.017         \\
                    &      [.007]         &      [.007]         &      [.012]         &      [.007]         &      [.007]         &      [.014]         \\
Number of children  &                     &       -.853\sym{***}&       -.885\sym{***}&                     &       -.957\sym{***}&       -.990\sym{***}\\
                    &                     &      [.197]         &      [.133]         &                     &      [.172]         &      [.179]         \\
Age 20-30, \%       &                     &      -1.783\sym{*}  &      -3.424\sym{***}&                     &      -2.303\sym{**} &      -3.600\sym{***}\\
                    &                     &      [.993]         &      [.475]         &                     &     [1.022]         &      [.580]         \\
Age 31-40, \%       &                     &       2.123\sym{**} &        .589         &                     &       1.830\sym{**} &        .368         \\
                    &                     &      [.940]         &      [.554]         &                     &      [.874]         &      [.677]         \\
Age 41-50, \%       &                     &       -.464         &      -1.196\sym{*}  &                     &       -.657         &      -1.221         \\
                    &                     &      [.905]         &      [.623]         &                     &      [.927]         &      [.760]         \\
Age 50+             &                     &       -.373         &      -1.899\sym{***}&                     &       -.999         &      -2.322\sym{***}\\
                    &                     &      [.616]         &      [.320]         &                     &      [.601]         &      [.417]         \\
Male, \%            &                     &       -.066         &        .098\sym{**} &                     &       -.019         &        .137\sym{**} \\
                    &                     &      [.070]         &      [.048]         &                     &      [.074]         &      [.057]         \\
Married, \%         &                     &       -.000         &        .201         &                     &        .436         &        .697         \\
                    &                     &      [.528]         &      [.345]         &                     &      [.475]         &      [.502]         \\
White, \%           &                     &        .117         &       -.682\sym{***}&                     &       -.178         &       -.937\sym{***}\\
                    &                     &      [.308]         &      [.180]         &                     &      [.280]         &      [.210]         \\
Black, \%           &                     &        .168         &       -.345\sym{*}  &                     &        .053         &       -.490\sym{**} \\
                    &                     &      [.255]         &      [.182]         &                     &      [.240]         &      [.240]         \\
College degree, \%  &                     &        .644\sym{***}&        .042         &                     &        .488\sym{*}  &       -.074         \\
                    &                     &      [.235]         &      [.122]         &                     &      [.253]         &      [.181]         \\
Masters degree, \%  &                     &        .007         &        .177         &                     &        .176         &        .382         \\
                    &                     &      [.565]         &      [.312]         &                     &      [.578]         &      [.407]         \\
PhD degree, \%      &                     &       -.492         &      -2.811\sym{***}&                     &      -1.375         &      -4.020\sym{***}\\
                    &                     &     [2.333]         &     [1.055]         &                     &     [2.404]         &     [1.461]         \\
Lives in metro, \%  &                     &       -.514         &       -.644\sym{***}&                     &       -.507         &       -.606\sym{***}\\
                    &                     &      [.357]         &      [.163]         &                     &      [.325]         &      [.207]         \\
Reading skills      &                     &                     &        .061\sym{***}&                     &                     &        .068\sym{***}\\
                    &                     &                     &      [.014]         &                     &                     &      [.021]         \\
Active listening skills&                     &                     &        .057\sym{***}&                     &                     &        .048\sym{***}\\
                    &                     &                     &      [.011]         &                     &                     &      [.014]         \\
Critical thinking skills&                     &                     &       -.030         &                     &                     &       -.039         \\
                    &                     &                     &      [.029]         &                     &                     &      [.040]         \\
Problem solving skills&                     &                     &       -.035         &                     &                     &       -.040         \\
                    &                     &                     &      [.032]         &                     &                     &      [.038]         \\
Judgment skills     &                     &                     &       -.044         &                     &                     &       -.022         \\
                    &                     &                     &      [.046]         &                     &                     &      [.053]         \\
Systems analysis skills&                     &                     &        .041\sym{***}&                     &                     &        .032         \\
                    &                     &                     &      [.013]         &                     &                     &      [.022]         \\
Operations analysis skills&                     &                     &       -.011         &                     &                     &       -.012         \\
                    &                     &                     &      [.010]         &                     &                     &      [.012]         \\
R-squared           &         .32         &         .83         &         .95         &         .42         &         .86         &         .94         \\
Sample Size         &          71         &          71         &          71         &          71         &          71         &          71         \\
Demographic Controls&          No         &          No         &          No         &          No         &          No         &          No         \\
Skill Controls      &          No         &         Yes         &         Yes         &          No         &         Yes         &         Yes         \\